\documentclass[aps,pra,noeprint,twocolumn,superscriptaddress,balancelastpage,10pt]{revtex4-1}
\usepackage{amsthm}
\newtheorem{theorem}{Theorem}
\newtheorem{corollary}{Corollary}

\usepackage{newtxtext}
\usepackage[utf8]{inputenc}
\usepackage{amsmath,amssymb,amsfonts,stmaryrd}
\usepackage{physics}
\usepackage{dsfont}
\usepackage{graphicx}
\usepackage[dvipsnames]{xcolor}
\usepackage{natbib}
\usepackage[linktocpage]{hyperref}
\hypersetup{colorlinks=true,citecolor=blue,linkcolor=BrickRed,urlcolor=Blue,breaklinks=true}

\usepackage{bm} 
\usepackage[margin=0.75in]{geometry}
\usepackage[capitalize]{cleveref}

\usepackage[english]{babel}

\makeatletter
\def\bbl@set@language#1{%
  \edef\languagename{%
    \ifnum\escapechar=\expandafter`\string#1\@empty
    \else\string#1\@empty\fi}%
  \@ifundefined{babel@language@alias@\languagename}{}{%
    \edef\languagename{\@nameuse{babel@language@alias@\languagename}}%
  }%
  \select@language{\languagename}%
  \expandafter\ifx\csname date\languagename\endcsname\relax\else
    \if@filesw
      \protected@write\@auxout{}{\string\select@language{\languagename}}%
      \bbl@for\bbl@tempa\BabelContentsFiles{%
        \addtocontents{\bbl@tempa}{\xstring\select@language{\languagename}}}%
      \bbl@usehooks{write}{}%
    \fi
  \fi}
\newcommand{\DeclareLanguageAlias}[2]{%
  \global\@namedef{babel@language@alias@#1}{#2}%
}
\makeatother
\DeclareLanguageAlias{en}{english}

\Crefname{equation}{Eq.}{Eqs.}
\Crefname{figure}{Fig.}{Figs.}
\Crefname{tabular}{Tab.}{Tabs.}
\Crefname{section}{Sec.}{Secs.}

\usepackage{braket}
\DeclareMathOperator{\Var}{Var}

\newcommand{\fe}{\mathsf{FE}}
\newcommand{\ce}{\mathsf{CE}}
\newcommand{\lfe}{\mathcal{L}_{\fe}}
\newcommand{\lce}{\mathcal{L}_{\ce}}

\renewcommand{\i}{\mathrm{i}}
\newcommand{\e}{\mathrm{e}}

\addto\captionsenglish{}

\begin{document}

\title{Recovering optimal precision in quantum sensing using imperfect control}
\author{Zi-Shen Li}\thanks{These authors contributed equally.}
\affiliation{QICI Quantum Information and Computation Initiative, Department of Computer Science, School of Computing and Data Science, The University of Hong Kong, Pokfulam Road, Hong Kong, China}

\author{Xinyue Long}\thanks{These authors contributed equally.}
\affiliation{Quantum Science Center of Guangdong-Hong Kong-Macao Greater Bay Area (Guangdong), Shenzhen, 518045, China}

\author{Xiaodong Yang}
\email{yangxd@szu.edu.cn}
\affiliation{Institute of Quantum Precision Measurement, State Key Laboratory of Radio Frequency Heterogeneous Integration, College of Physics and Optoelectronic Engineering, Shenzhen University, Shenzhen 518060, China}
\affiliation{Quantum Science Center of Guangdong-Hong Kong-Macao Greater Bay Area (Guangdong), Shenzhen, 518045, China}

\author{Yuxiang Yang}
\email{yuxiang@cs.hku.hk}
\affiliation{QICI Quantum Information and Computation Initiative, Department of Computer Science, School of Computing and Data Science, The University of Hong Kong, Pokfulam Road, Hong Kong, China}

\author{Dawei Lu}
\email{ludw@sustech.edu.cn}
\affiliation{Shenzhen Institute for Quantum Science and Engineering and Department of Physics, Southern University of Science and Technology, Shenzhen, 518055, China}
\affiliation{Quantum Science Center of Guangdong-Hong Kong-Macao Greater Bay Area (Guangdong), Shenzhen, 518045, China}
\affiliation{Guangdong Provincial Key Laboratory of Quantum Science and Engineering, Southern University of Science and Technology, Shenzhen 518055, China}

\begin{abstract}
Quantum control plays a crucial role in enhancing precision scaling for quantum sensing. 
However, most existing protocols require perfect control, even though real-world devices inevitably have control imperfections.
Here, we consider a fundamental setting of quantum sensing with imperfect clocks, where the duration of control pulses and the interrogation time are all subject to uncertainty.
Under this scenario, we investigate the task of frequency estimation in the presence of a non-Markovian environment.
We design a control strategy and prove that it outperforms any control-free strategies, recovering the optimal Heisenberg scaling up to a small error term that is intrinsic to this model. 
We further demonstrate the advantage of our control strategy via experiments on a nuclear magnetic resonance (NMR) platform. 
Our finding confirms that the advantage of quantum control in quantum sensing persists even in the presence of imperfections. 
\end{abstract}

\maketitle

\textit{Introduction.---}Quantum sensors that harness coherence and entanglement can surpass the ultimate classical limit of precision \cite{giovannettiAdvancesQuantumMetrology2011,degenQuantumSensing2017}.
However, this quantum advantage is often vulnerable to various types of noise \cite{huelgaImprovementFrequencyStandardsQuantumEntanglement1997,dornerOptimalQuantumPhaseEstimation2009,escherGeneralFrameworkEstimatingUltimatePrecisionLimit2011,demkowicz-dobrzanskiElusiveHeisenbergLimitQuantumenhancedMetrology2012,kolodynskiEfficientToolsQuantumMetrologyUncorrelatedNoise2013,thomas-peterRealworldQuantumSensorsEvaluatingResourcesPrecision2011}.
Previous research has demonstrated that the advantage can be restored with the assistance of quantum control, including error correction \cite{kesslerQuantumErrorCorrectionMetrology2014,durImprovedQuantumMetrologyUsingQuantumError2014,arradIncreasingSensingResolutionErrorCorrection2013,demkowicz-dobrzanskiAdaptiveQuantumMetrologyGeneralMarkovianNoise2017,zhouAchievingHeisenbergLimitQuantumMetrologyUsing2018}, dynamical decoupling \cite{sekatskiDynamicalDecouplingLeadsImprovedScalingNoisy2016}, learning control strategies \cite{poggialiOptimalControlOnequbitQuantumSensing2017,liuQuantumParameterEstimationOptimalControl2017,marciniakOptimalMetrologyProgrammableQuantumSensors2022,fallaniLearningFeedbackControlStrategiesQuantumMetrology2022,zhaiControlenhancedQuantumMetrologyMarkovianNoise2023,yangControlenhancedNonMarkovianQuantumMetrology2024}, and Floquet engineering \cite{baiFloquetEngineeringOvercomeNoGoTheoremNoisy2023}.
Most existing applications assume perfect control, whereas real-world devices are subject to imperfections that can undermine the advantages of control-assisted quantum sensing.
For example, the lagging effect of error correction can lead to additional bias \cite{rojkovBiasErrorcorrectedQuantumSensing2022} and the uncertainty of control fields challenges the utility of critical metrology \cite{mihailescuUncertainQuantumCriticalMetrologySingleMulti2024a}.
It is unknown whether control protocols assisted by imperfect control operations can still achieve superior performance in quantum sensing.


In this Letter, we consider quantum sensing with devices subject to clock uncertainty.
Explicitly, we assume that the duration of \textit{any} operation performed by the experimenter, such as the preparation of the probe state, the intermediate control, the interrogation process, or the final measurement, is subject to a small uncertainty.
Consequently, the measurement statistics of quantum system are averaged in time. 
More broadly, clock uncertainty encompasses unknowable discrepancies between the experimenter's prior knowledge of frequency or time and their actual values in reality.
From a fundamental perspective, the time averaging reduces the off-diagonal terms in quantum state, thereby making it relevant to dephasing, \cite{xuerebImpactImperfectTimekeepingQuantumControl2023,ballRoleMasterClockStabilityQuantumInformation2016} ergodicity \cite{penroseFoundationsStatisticalMechanics1979,deutschEigenstateThermalizationHypothesis2018}, and thermalization \cite{srednickiChaosQuantumThermalization1994,closTimeresolvedObservationThermalizationIsolatedQuantumSystem2016}.

We show, both theoretically and experimentally, that protocols utilizing imperfect control can still offer advantages over control-free protocols in the presence of clock uncertainty. 
Specifically, we consider the task of frequency estimation under a non-Markovian environment with strong system-environment coupling.
We demonstrate that clock uncertainty can significantly compromise both precision and accuracy in the control-free case.
To address this issue, we propose a control strategy and show that its error is substantially reduced compared to any strategy without intermediate control.
Remarkably, we observe that the error of our control strategy nearly matches the \textit{hardware limit}, which we define to be the systematic error that persists even in the absence of quantum decoherence due to clock uncertainty.
In this sense, our strategy recovers the optimal precision of quantum sensing.
Our method is not a simple extension of any existing approaches, such as dynamical decoupling \cite{sekatskiDynamicalDecouplingLeadsImprovedScalingNoisy2016} or error correction \cite{kesslerQuantumErrorCorrectionMetrology2014,durImprovedQuantumMetrologyUsingQuantumError2014,arradIncreasingSensingResolutionErrorCorrection2013,demkowicz-dobrzanskiAdaptiveQuantumMetrologyGeneralMarkovianNoise2017,zhouAchievingHeisenbergLimitQuantumMetrologyUsing2018}. 
We further demonstrate our findings via experiments on a nuclear magnetic resonance (NMR) processor, where our control strategy is confirmed to be superior to control-free strategies.  
Our results confirm that quantum control with imperfections can still bring advantages to quantum sensing.

\medskip

\textit{Frequency estimation in the presence of clock uncertainty.}---We consider a frequency estimation task in an open quantum system coupled to an environment.
The Hamiltonian of the joint system is $H(\omega)=H_{\mathsf{S}}(\omega) + H_{\mathsf{E}} + H_{\mathsf{SE}}$, where $H_{\mathsf{S}}(\omega)$ is the system Hamiltonian that depends on a frequency $\omega$, $H_{\mathsf{E}}$ is the environment Hamiltonian, and $H_{\mathsf{SE}}$ is a system-environment interaction term.
The task is to estimate the frequency $\omega$ by interrogating the system for time $T$ in each run of the experiment.
We consider a model of clock uncertainty where the stopwatch used by the experimenter to determine the time interval between any pair of events (e.g., the start and the end of interrogation) is inaccurate. As illustrated in \Cref{fig1}~(a), time intervals measured by the stopwatch are subject to some inherent clock uncertainty.
As shown in \Cref{fig1}~(b) and (c), clock uncertainty affects each time slot, including the pulse widths and the intervals between pulses.
Mathematically, we model the clock uncertainty as a map that takes any positive real number $t$ (the duration of an operation set by the experimenter) to $\hat{t}=t(1+u)$ (the actual duration of the operation), where $u$ is a random variable with a probability distribution $f$. 
Therefore, the clock uncertainty is captured by the function $f$, which is unknown to the experimenter and assumed to be bounded, i.e., its domain is a subset of $[-\epsilon,\epsilon]$ for some $\epsilon\ge0$.

To estimate $\omega$, the experimenter is given access to an ancilla and begins by preparing the system and the ancilla in a suitable probe state. The system evolves for time $T$ with coupling to the environment under the Hamiltonian $H(\omega)$ and then gets measured jointly with the ancilla. 
During the evolution, the experimenter is allowed to perform {\it intermediate control}, which could be either pulse control or continuous control. This procedure is repeated for $\nu$ times to generate measurement statistics and to output an estimate $\hat{\omega}$ in the end. 
A configuration of the state preparation, measurement, and intermediate control is referred to as a {\it strategy}. Here we focus on the comparison between two families of strategies: free-evolution (FE) strategies, which do not use any intermediate control, and the generic family of all control-enhanced (CE) strategies.  
For comparison, we define the following loss function:
\begin{align}\label{eq:loss-max-min-mse}
    \mathcal{L_S} =  \min_{s \in \mathcal{S}} \max_{f \in \mathcal{N}} \mathcal{R}\left(f, s\right),
\end{align}
where $\mathcal{R}\left(f, s\right)$ is the mean-squared error (MSE) of the estimate corresponding to a strategy $s$ under the clock uncertainty distribution $f$, $\mathcal{S}\in\{\fe,\ce\}$, and $\mathcal{N}$ is the set of all possible uncertainty distributions (known to the experimenter).
Intuitively, the experimenter has to design the best strategy while accounting for the error corresponding to any clock uncertainty distribution $f\in\mathcal{N}$, which is bounded by $\mathcal{L_S}$.

For each repetition of the experiment, the measurement outcome obeys a distribution $p(\omega,s,u_1,u_2,\dots)$ that depends on the unknown parameter $\omega$, the strategy $s$, and the clock uncertainty $u_i$ ($i=1,2,\dots$) of the $i$-th operation. Since each $u_i$ follows the distribution $f$, the effective distribution of the outcome is $p_f(\omega,s)=\int du_1\,f(u_1)\int du_2\,f(u_2)\cdots\,p(\omega,s,u_1,u_2,\dots)$. The estimate is generated from $\nu$ independent and identically-distributed samples with the distribution $p_f(\omega,s)$.
The estimation error is thus bounded by the Cram\'er-Rao bound (CRB). 
Note that the estimate does not necessarily satisfy the unbiasedness condition due to the clock uncertainty. Therefore, we need to use the general form of the CRB (see Ref.~\cite{treeDetectionEstimationModulationTheory2001} for more details)
\begin{align}\label{eq:bias-crb}
    \mathcal{R}\left(f, s\right)\ge \frac{\left(1+\partial b/\partial \omega\right)^2}{\nu\mathcal{F}_{\omega}(f,s)}+b(f,s)^2,
\end{align}
where $\mathcal{F}_{\omega}(f,s)$ is the Fisher information (FI) of $p_f(\omega,s)$ at $\omega$, and $b$ is the bias of the estimator defined as $b:=\mathrm{E}(\hat{\omega})-\omega$.
The FI is upper bounded by the quantum Fisher information (QFI) of the system state prior to the measurement, denoted as $\mathcal{F}^{\rm Q}_{\omega}(f,s)$, which can be achieved by optimizing the measurement in the single-parameter case \cite{braunsteinStatisticalDistanceGeometryQuantumStates1994}.

Note that the CRB (\ref{eq:bias-crb}) is not necessarily achievable. Instead, we will use it to set a lower bound $\mathcal{L}_{\fe}$, the loss function of the free-evolution case. On the other hand, we will provide an upper bound for $\mathcal{L}_{\ce}$ by designing a specific control strategy and estimating its error. 



\medskip

\textit{The advantage of control.}---We compare FE and CE through a concrete example of noisy frequency estimation. Let both the system and the environment be two-dimensional and $H_{\mathsf{E}}$ be the identity, $H_{\mathsf{S}}(\omega)=(\omega/2) \sigma_Z^{\mathsf{S}}$ with $\omega$ being the parameter of interest, and $H_{\mathsf{SE}}=g\text{SWAP}$, where $\sigma_Z^{\mathsf{S}}$ denotes the Pauli-$Z$ rotation on the system, $g\ge 0$ is the interaction strength and SWAP denotes the SWAP gate.
The environment is initialized in the state $\ket{0^E}$. 
The set of possible clock uncertainty distributions is defined as $\mathcal{N}=\{f|f(u)=0~\forall~u\notin [-\epsilon, \epsilon]\}$.
For simplicity, we assume for now $f\in\mathcal{N}_\delta$ to be a delta function $\mathcal{N}_{\delta}:=\{\delta(u-\xi)|\xi\in[-\epsilon,\epsilon]\}$ and discuss later how this assumption can be lifted.

In the absence of the environment [here referred to as the \emph{interaction-free} (IF) case, i.e., when $g=0$], the statistical error is bounded by the inverse of the quantum Fisher information (QFI), which is $1/(\nu T^2)$, known as the \textit{Heisenberg limit} \cite{giovannettiQuantumMetrology2006}. As the total interrogation time is still subject to the clock uncertainty, there is a systematic error of the scale $O(\epsilon^2\omega^2)$ \cite{SM}.
Since this systematic error term persists in the absence of the environment, we define it as the \textit{hardware limit} of this task.

Next, we introduce the FE case, where the joint system undergoes the evolution $U_{\mathsf{FE}}(T):=e^{-iHT}$. Strategies $s\in\fe$ differ only in probe states and measurements [corresponding to the pulses depicted in \Cref{fig1}~(b)]. Our analysis focuses on the strong-coupling regime, where $g\gg|\omega|$. 
Subsequently, we show a lower bound on the loss function for any FE strategy.

In order to bound $\lfe$, defined in Eq.~(\ref{eq:loss-max-min-mse}), we first apply the max-min inequality, which yields $\lfe\ge\max_{f\in\mathcal{N}}\min_{s\in\fe}\mathcal{R}(f,s)$. Next, we apply the (biased) CRB (\ref{eq:bias-crb}) to bound the MSE of every FE strategy. Explicitly, we show that (see Ref.~\cite{SM} 
for the proof) 
\begin{equation}
	\begin{aligned}\label{eq:inter1}
        \mathcal{F}^{\rm Q}_{\omega}\left[\delta(u-\xi),\fe\right] \le& \frac{2T^2(1+\xi)^2 \cos^2 \left[T \Omega (1+\xi)\right]}{\cos \left[2 T(1+\xi) \Omega \right]+3} \\
        &+ O\left({\omega^2T^2}/{g^2}\right),
	\end{aligned}
\end{equation}
where $\mathcal{F}^{\rm Q}_{\omega}\left[\delta(u-\xi),\fe\right]$ is the QFI maximized over all FE strategies with delta function $f$ and $\Omega:=\sqrt{g^2+\omega^2/4}$.
It remains to be determined how the bias can affect the CRB: 
As shown in Eq.~(\ref{eq:bias-crb}), the statistical error can be reduced if $\partial b/\partial\omega$ is negative, at the cost of increasing the systematic error.
Therefore, we can lower-bound the total MSE by taking the optimal tradeoff between these two error terms:
\begin{align}\label{eq:inter2}
    \lfe \ge \max_{f\in\mathcal{N}_\delta} \frac{1}{\omega^{-2} + \nu \mathcal{F}_\omega^Q\left[\delta(u-\xi),\fe\right]}.
\end{align}
We assume that the value of $\omega$ is close to a specific reference value, which we set to zero in this case, i.e., $|\omega| \ll 1$. 
This assumption is necessary to ensure that $\omega T$ falls within a known period, satisfying the local condition of the CRB theorem \cite{explain-assumption-opt-bias}.
Note that saturating \Cref{eq:inter2} necessitates arbitrarily changing the bias, which is generally unrealistic.
Finally, combing Eqs. (\ref{eq:inter1}) and (\ref{eq:inter2}) and taking the maximum over $\mathcal{N}_\delta$,
we get the following lower bounds on the loss function within small corrections \cite{corrections-loss-FE}:
\begin{align}
    &\lfe \ge \frac{g^2\omega^2}{\nu\omega^4T^2 + g^2}, &\text{ if } \epsilon \ge \epsilon^\ast,\label{eq:worst-case-failed-omega}\\
    &\lfe \ge \frac{1}{\nu  T^2 \left[\frac{\omega ^2}{g^2}+\cos ^2(\epsilon T \Omega )\right]+\frac{1}{\omega ^2}}, &\text{ if } \epsilon < \epsilon^\ast,\label{eq:worst-case-succ-omega}
\end{align}
where $\epsilon^\ast :=\pi/(2\Omega T)$.
Note that the bound (\ref{eq:worst-case-failed-omega}) scales as $g^2/(\nu\omega^2T^2)$ when $T$ and $\nu$ are sufficiently large, which is a factor of $g^2/\omega^2$ worse than the Heisenberg scaling $1/(\nu T^2)$ in the strong-coupling limit $g\gg\omega$. Intuitively, to be close to the Heisenberg scaling, the system needs to be decoupled from the environment before the measurement. However, if the clock uncertainty $\epsilon$ is large enough so that $T\epsilon$ scales as the period of evolution, which is approximately $\Theta(1/g)$, the experimenter will not be able to choose a better-than-random timing of decoupling, rendering the final estimate to be much less accurate. 
We further remark that the assumption on $f$ to be a delta function can be lifted by using the convexity of the QFI, and the bound (\ref{eq:worst-case-failed-omega}) holds for generic $f$.

\begin{figure}
    \centering
    \includegraphics[width=0.95\linewidth]{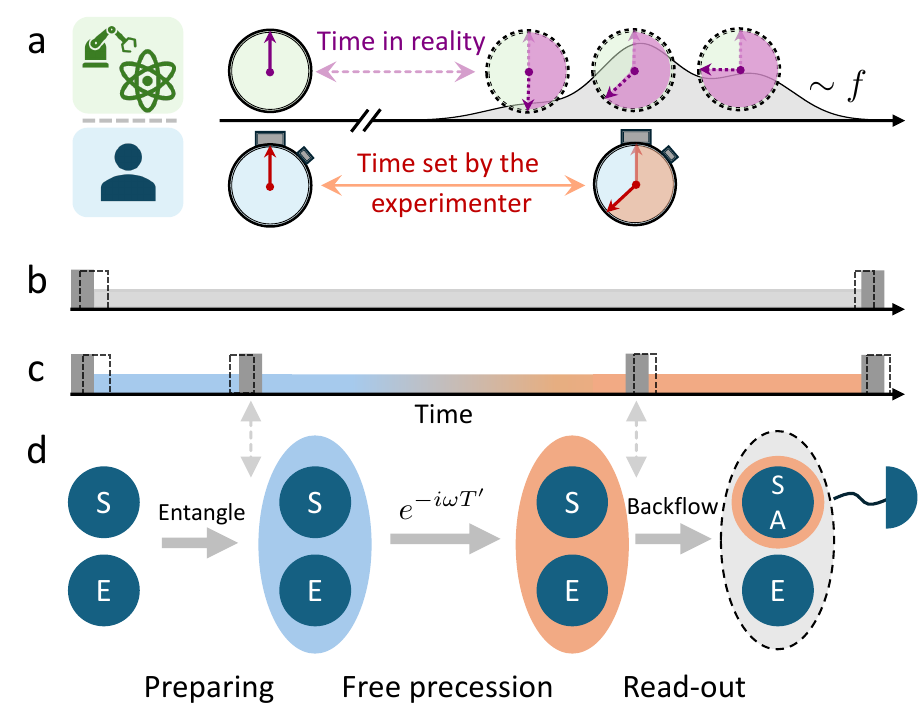}
    \caption{\textbf{The clock uncertainty model and the control strategy.} We consider a generic model of clock uncertainty where any time interval, set by the experimenter using an imperfect stopwatch, is subject to some random deviation [see \textbf{(a)}]. 
    Consequently, the total interrogation time [see \textbf{(b)}], the duration of control pulses and the length of the interval between any consecutive pulses [\textbf{(c)}] are all subject to uncertainty.
    Under this model, we propose a control strategy to enhance the performance of frequency estimation [see \textbf{(d)}], where the system (S) is coupled to an environment (E) and an ancillary qubit (A). The dashed arrows between \textbf{(c)} and \textbf{d)} indicate the timing of inserting two intermediate control pulses (grey boxes in the middle). 
    }
    \label{fig1}
\end{figure}

Now, we consider the CE case, where intermediate control is allowed.
We show that the Heisenberg scaling of the statistical error can be recovered by using a few intermediate control pulses \cite{explain-fast-ctrl} and one ancillary qubit. As shown in \Cref{fig1}~(c) and (d), the process involves four control pulses. A first pulse is applied to initialize the system. After the system-environment entanglement is established, a second pulse is applied to the system, mapping the system-environment state into a suitable subspace that is sensitive only to the evolution induced by $H_{\mathsf{S}}(\omega)$. Towards the end of the interrogation, a third pulse is applied to the joint system of the system and the ancilla, facilitating the information backflow from the environment to the joint system. At last, a fourth pulse, combined with a measurement in the computational basis, is applied to realize the final measurement.
A more detailed explanation will be provided later. The performance of our method is summarized in the following theorem:

\begin{theorem}\label{thm:MSE-CE}
Assuming the clock uncertainty to be bounded, i.e., $\mathcal{N}=\{f|f(u)=0~\forall~u\notin [-\epsilon, \epsilon]\}$, and $\nu$ to be sufficiently large, the loss function of the CE case, defined by \Cref{eq:loss-max-min-mse}, is upper bounded as:
    \begin{align}\label{eq:MSE-CE}
        \lce \le \frac{1+O(\eta)}{\nu T'^2} + \frac{16\csc^2(\omega T')\eta^2+O(\eta^3)}{T'^2},
    \end{align}
    where $\eta={3\pi|\omega|}/({8g})+{11\pi\epsilon}/{4}+{\epsilon|\omega| T'}/{2}$ and $T':=T-{3\pi}/{(4g)}$.
\end{theorem}
The proof can be found in Ref.~\cite{SM}. 
From Theorem \ref{thm:MSE-CE}, it is clear that the variance of the estimate of the CE case [the first term in \Cref{eq:MSE-CE}] coincides with that of the IF case in the leading order, which is significantly smaller than that of the FE case.
Note that we do not assume the control pulses to be perfect:
The clock uncertainty will influence the length of the time intervals between the control pulses, as well as the pulses' length [as illustrated in \Cref{fig1}~(b) and (c)]. As a consequence, the final estimate of the CE case will also suffer from a small bias, which is nevertheless close to the hardware limit $O(\epsilon^2\omega^2)$ up to a constant factor.

For illustration, we run a numerical simulation of this task for FE and CE cases.
The parameters are configured to  $\omega=1/300, g=10, T=80\pi,\nu =10^4$ and the function $f$ is set to a uniform distribution in the interval $[-\epsilon,\epsilon]$.
We denote by $\mathcal{R}_{\fe}$ the MSE of the FE strategy that is optimal in the absence of clock uncertainty ($\epsilon=0$) \cite{explain-epsilon0-optimal} and by $\mathcal{R_{\ce}}$ the MSE of our CE strategy.
To see the reliability of the strategies, we evaluate the performances of the strategies for different $\omega$ sampled uniformly from the interval $[1/500,1/100]$.
For each $\epsilon$, the standard deviation of the MSEs for $10^2$ samples corresponding to the FE (CE) strategy is calculated, shown as the red (blue) shaded area in \Cref{fig2}~(a).
The CE strategy is illustrated in \Cref{fig2}~(b).
From \Cref{fig2}~(a), it is clear that $\mathcal{R}_\fe$ is lower than $\mathcal{R}_\ce$ by several magnitudes. 
The gap is considerably larger than the prediction of the theoretical bounds (\ref{eq:worst-case-failed-omega}) and (\ref{eq:MSE-CE}).
This is because the bias comprises a significant portion of the total MSE of the FE strategy, which is not properly reflected in the bound (\ref{eq:worst-case-failed-omega}).
The MSE of the IF case is set to $1/(\nu T^2)+\epsilon^2\omega^2$ in the plot \cite{SM}.
Remarkably, the MSE of our CE strategy is close to that of the IF case.
This justifies the optimality of the control, as it nearly eliminates the impact of the environment.

The underlying reasons why our control strategy is effective are the following. The system is first prepared in a state so that it gets entangled with the environment through their interaction [the first step in \Cref{fig1}~(d)]. 
We then apply a local control pulse to send the state into the stable subspace of the interaction Hamiltonian. This subspace, spanned by $\{\ket{0^S0^E},\ket{1^S1^E}\}$, is a {degenerate eigen-subspace} (DES) of the interaction Hamiltonian. 
It is noteworthy that the states within the DES can be entangled, distinguishing it from the decoherence-free subspace \cite{lidarDecoherenceFreeSubspacesQuantumComputation1998}.
Meanwhile, entangled states in this DES, such as $(\ket{0^S0^E}+\ket{1^S1^E})/\sqrt{2}$, can evolve within the subspace under the system Hamiltonian [the second step in \Cref{fig1}~(d)].
During the subsequent free precession, the frequency $\omega$ is encoded as a phase factor with maximal sensitivity, which is nevertheless not locally measurable. Therefore,
before the measurement, we apply another control to transition the state out of the DES, allowing information to flow back to the system [the third step in \Cref{fig1}~(d)].
In contrast, in the FE case, the system Hamiltonian $(\ketbra{0^S}-\ketbra{1^S})\otimes \mathds{1}^E$ will be reduced to $\ketbra{0^S 0^E}-\ketbra{1^S 1^E}$ by the \textit{Zeno dynamics} \cite{facchiQuantumZenoDynamicsMathematicalPhysicalAspects2008,burgarthGeneralizedAdiabaticTheoremStrongCouplingLimits2019}, leading to information loss even in the absence of clock uncertainty \cite{classical-interpretation-SWAP}. Our control strategy uses the DES to effectively mitigate this effect, resulting in a performance enhancement.

\begin{figure}
    \centering
    \includegraphics[width=0.95\linewidth]{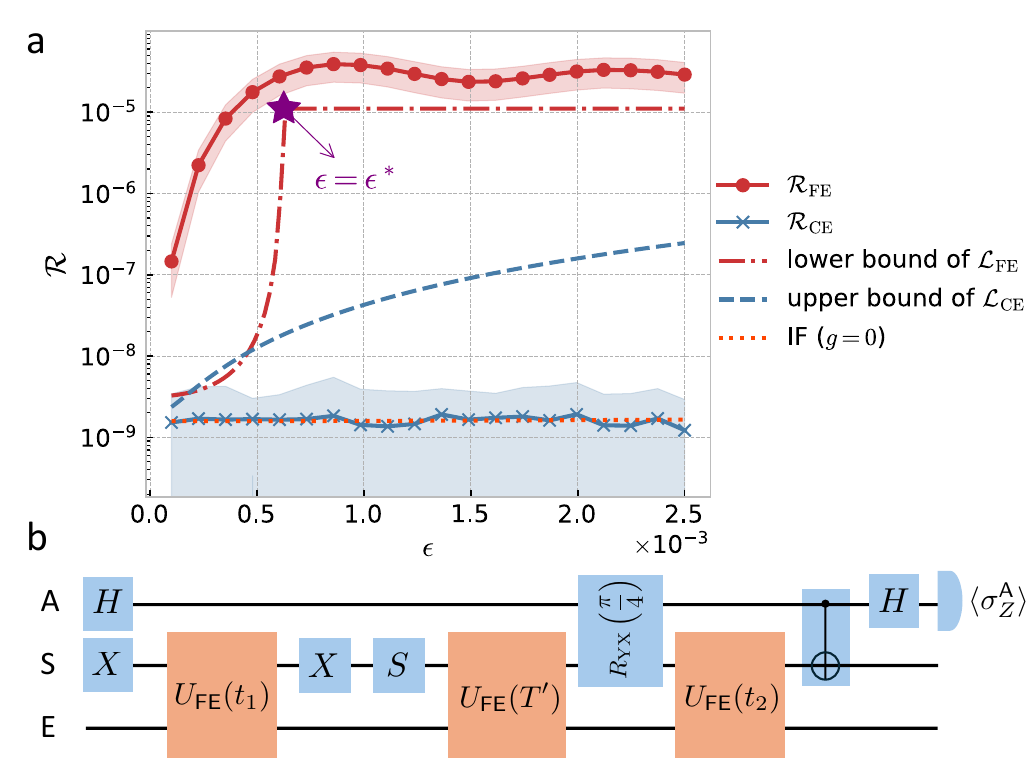}
    \caption{\textbf{Numerical results for finite $\nu$ and the CE strategy.} \textbf{(a)} Numerical results for both the FE case and the CE case with a uniform $f$: The IF case corresponds to setting the system-environment interaction strength to zero. The heights of the shaded areas indicate the standard deviation resulting from different choices of $\omega$. Note that in some instances, the error in the CE case is even smaller than in the IF case, which can be attributed to the CE case exhibiting a larger bias compared to the IF scenario that accidentally benefits the estimation in the non-asymptotic regime. \textbf{(b)} The circuit representation of our CE strategy: The gates $H,~S,$ and$~X$ are the Hadamard gate, the $\pi/2$-phase gate, and the Pauli-$X$ gate, respectively. The unitary $U_{\mathsf{FE}}$ is the free evolution under the total Hamiltonian. The lengths of the intervals between the control operations are $t_1=\pi/(4g)$ and $t_2=\pi/(2g)$, the initial state is $\ket{000}$,
    and the rotation gate is defined as $R_{\rm YX}(\theta):=\exp(-i\sigma_Y^{\mathsf{A}}\sigma_X^{\mathsf{S}}\theta)$ with $\sigma_P^{\mathsf{A}~(\mathsf{S})}$ denoting the Pauli-$P$ operator on the ancilla (system).
    After the measurement, we use the relation $\braket{\sigma_{Z}^{\mathsf{A}}}=\cos(\omega T'/2)$ to determine the final estimate of $\omega$.}
    \label{fig2}
\end{figure}

\medskip

\textit{Experimental implementation.}--- To demonstrate the advantage of our CE strategy in practice, we conduct experiments with NMR systems.
An important feature of NMR systems is that observations are made on large ensembles rather than individual quantum states. For quantum sensing, the same strategy will be performed simultaneously on a large number (close to the Avogadro constant) of identical systems. This is effectively to consider the performance of the strategies in the large repetition limit, i.e., $\nu\to\infty$, where the systematic error becomes the only prominent term. We remark that this scenario is complementary to that of the previous numerical simulation, which focuses on the finite $\nu$ case where both the systematic error and the statistical error have nontrivial impact on the performance. 

\begin{figure}
    \centering
    \includegraphics[width=0.75\columnwidth]{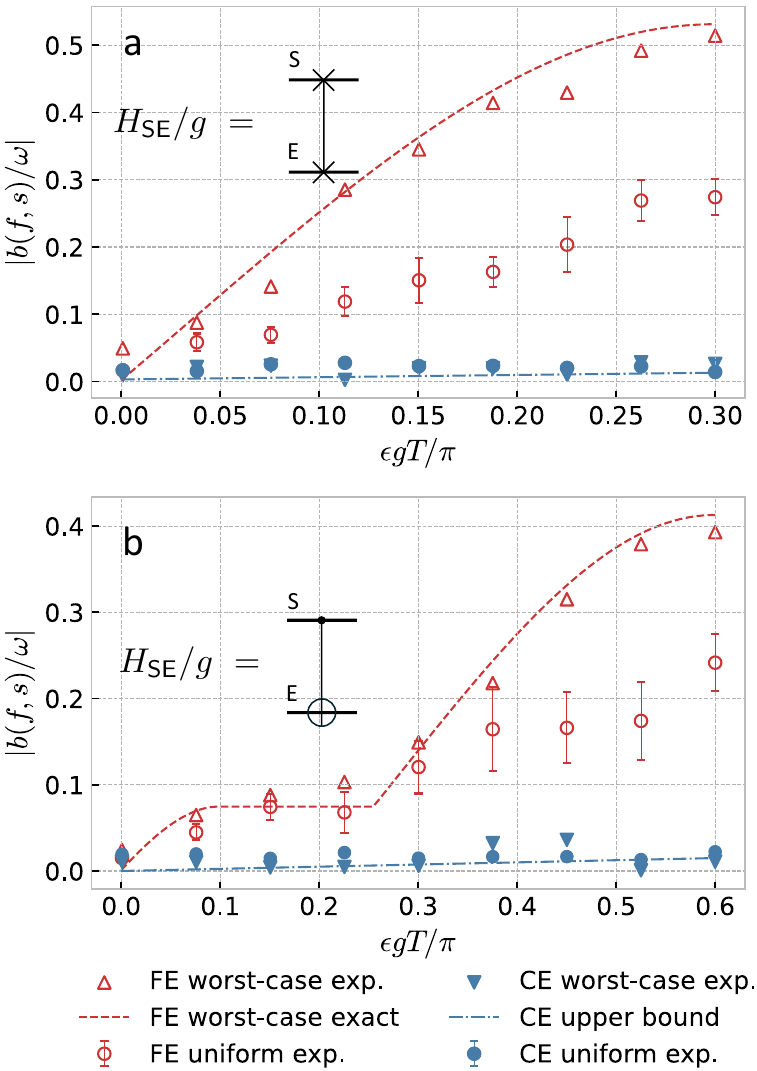}
    \caption{\textbf{Experimental demonstration on an NMR platform.} In \textbf{(a)} and \textbf{(b)}, the relative errors are plotted as functions of the clock uncertainty for the SWAP interaction and the CNOT interaction, respectively. Here the label \textsf{S} (\textsf{E}) denotes the system (environment). We consider both when the clock uncertainty distribution is fixed to be the uniform distribution over $[-\epsilon,\epsilon]$ and when it can be any bounded distribution over this interval. The length of the errorbar indicates the standard error in 10 repetitions of the experiment. }
    \label{fig:experiment-result}
\end{figure}

We conduct our experiments using a Bruker 600 MHz spectrometer at room temperature. 
The sample comprises $^{13}$C-labeled trans-crotonic acid dissolved in acetone-d$_6$, consisting of three $^{1}$H and four $^{13}$C nuclear spins, which collectively form a 7-qubit quantum processor. 
More details about this experimental platform are available in Ref.~\cite{longEntanglementEnhancedQuantumMetrologyColoredNoiseQuantum2022}.

In our experiment, we utilize the spins of three $^{13}$C nuclei --- C$_1$, C$_2$, and C$_3$ --- to simulate the ancilla, the system, and the environment, respectively. 
We configure the experimental parameters as $\omega=0.01~\text{kHz}$, $g=10~\text{kHz}$, and $T=80\pi~\text{ms}$. 
To explore different scenarios, we investigate two distinct interaction Hamiltonians: SWAP and CNOT \cite{explain-interaction-channel}, along with two different types of clock uncertainty.
The CE strategy in the SWAP interaction case aligns with the illustration presented in \Cref{fig2}~(b), while the one in the CNOT interaction case can be found in Ref.~\cite{SM}.
We first consider the uncertainty distribution to be fixed to the uniform distribution in the interval $[-\epsilon,\epsilon]$.
The error is measured by the relative bias $\left|{b(f,s)}/{\omega}\right|$, where $f$ is set to the uniform distribution.
We further consider the worst-case over all bounded clock uncertainty distributions, and the relative bias becomes 
\begin{align}\label{eq:exp-worst-case-MSE}
    \max_{f\in\mathcal{N}} \left|{b(f,s)}/{\omega}\right|,~s\in \fe \text{ or } \ce,
\end{align}
which is refered to as \textit{worst-case} error.
The strategy $s\in\fe$ is fixed to be the optimal one when $\epsilon=0$ \cite{explain-epsilon0-optimal}.
The outcomes of our experiments are presented in \Cref{fig:experiment-result}.
It is clear that the advantage is still significant in the demonstration for both types of interactions and both types of clock uncertainty, as shown in \Cref{fig:experiment-result}~(a) and (b).
We remark that the experimental finding complements the theoretical analysis, as the advantage of control cannot be manifested by comparing the theoretical bounds \Cref{eq:MSE-CE,eq:worst-case-failed-omega} in the $\nu \rightarrow \infty$ regime.
Note that in some instances of the CE case, the experimental error exceeds the theoretical upper bound. 
This can be attributed to imperfections of the experimental devices other than the clock uncertainty, such as the electromagnetic noise from electronic components, radiofrequency coils, and the associated circuits, which are not considered in the theoretical analysis. 
We emphasize that the overall error in the CE case remains relatively small despite all these additional imperfections.

\medskip

\textit{Conclusion and outlook.}---In this Letter, we consider frequency estimation under clock uncertainty that influences the entire sensing process.
We introduce an effective control strategy to mitigate the 
effects of clock uncertainty and recover the optimal sensing precision.
We prove and numerically verify the significant advantages of our CE strategy over FE strategies in the finite repetition limit ($\nu < \infty$).
We further validate our theory 
in the asymptotic regime ($\nu \rightarrow \infty$) on an NMR processor.

The distinction between the FE case and the CE case can also be interpreted using the ergodicity of quantum systems. In the FE case, the final state exhibits strong ergodicity. As a consequence, the time averaging induced by the clock uncertainty results in severe decoherence. 
In contrast, the final state in the CE case is much less sensitive to this time averaging, thanks to the ergodicity reduction imposed by the control operations. Ergodicity-breaking techniques have been an interesting direction of research studied recently in many-body quantum systems \cite{bernienProbingManybodyDynamics51atomQuantumSimulator2017,serbynQuantumManyBodyScarsWeakBreakingErgodicity2021}. Exploring how the connection with quantum sensing could be established may substantially extend the scope of application of our results to these systems.

In our present model, the noise effect of clock uncertainty has been treated semi-classically, whereas via the foundational models of autonomous quantum clocks \cite{malabarbaClockdrivenQuantumThermalEngines2015,erkerAutonomousQuantumClocksDoesThermodynamicsLimit2017,woodsAutonomousQuantumMachinesFiniteSizedClocks2019,woodsQuantumClocksAreMoreAccurateClassical2022,woodsAutonomousQuantumDevicesWhenAreThey2023,meierFundamentalAccuracyResolutionTradeTimekeepingDevices2023,yangUltimateLimitTimeSignalGeneration2020}, one can also treat this effect quantumly. Explicitly, the imperfect clocks can be modeled as quantum clocks built on finite-dimensional quantum systems, and the clock uncertainty becomes a result of the quantum correlation between the clock systems and the probe system, just as in the case of environment-induced decoherence.
This hints at the possibility of treating the problem of quantum sensing with imperfect control fully quantumly, using the framework of higher-order transformations \cite{chiribellaQuantumCircuitsArchitecture2008,chiribellaTransformingQuantumOperationsQuantumSupermaps2008,chiribellaTheoreticalFrameworkQuantumNetworks2009} and quantum comb metrology \cite{chiribellaOptimalNetworksQuantumMetrologySemidefinitePrograms2012,yangMemoryEffectsQuantumMetrology2019,altherrQuantumMetrologyNonMarkovianProcesses2021} to establish the ultimate precision limit.



\medskip

\textit{Acknowledgment.}---Z.-S. Li would like to acknowledge Prof. Mio Murao for valuable comments on the relationship between clock uncertainty and thermalization.
This work is supported by an Innovation Program for Quantum Science and Technology (project no.~2023ZD0300600), by the National Natural Science Foundation of China via project no.~12204230 and the Excellent Young Scientists Fund (Hong Kong and Macau) project no.~12322516, 
and the Hong Kong Research Grant Council (RGC) via the Early Career Scheme (ECS) grant 27310822 and the General Research Fund (GRF) grant 17303923.

\bibliography{mylib.bib}

\begin{thebibliography}{61}%
\makeatletter
\providecommand \@ifxundefined [1]{%
 \@ifx{#1\undefined}
}%
\providecommand \@ifnum [1]{%
 \ifnum #1\expandafter \@firstoftwo
 \else \expandafter \@secondoftwo
 \fi
}%
\providecommand \@ifx [1]{%
 \ifx #1\expandafter \@firstoftwo
 \else \expandafter \@secondoftwo
 \fi
}%
\providecommand \natexlab [1]{#1}%
\providecommand \enquote  [1]{``#1''}%
\providecommand \bibnamefont  [1]{#1}%
\providecommand \bibfnamefont [1]{#1}%
\providecommand \citenamefont [1]{#1}%
\providecommand \href@noop [0]{\@secondoftwo}%
\providecommand \href [0]{\begingroup \@sanitize@url \@href}%
\providecommand \@href[1]{\@@startlink{#1}\@@href}%
\providecommand \@@href[1]{\endgroup#1\@@endlink}%
\providecommand \@sanitize@url [0]{\catcode `\\12\catcode `\$12\catcode `\&12\catcode `\#12\catcode `\^12\catcode `\_12\catcode `\%12\relax}%
\providecommand \@@startlink[1]{}%
\providecommand \@@endlink[0]{}%
\providecommand \url  [0]{\begingroup\@sanitize@url \@url }%
\providecommand \@url [1]{\endgroup\@href {#1}{\urlprefix }}%
\providecommand \urlprefix  [0]{URL }%
\providecommand \Eprint [0]{\href }%
\providecommand \doibase [0]{http://dx.doi.org/}%
\providecommand \selectlanguage [0]{\@gobble}%
\providecommand \bibinfo  [0]{\@secondoftwo}%
\providecommand \bibfield  [0]{\@secondoftwo}%
\providecommand \translation [1]{[#1]}%
\providecommand \BibitemOpen [0]{}%
\providecommand \bibitemStop [0]{}%
\providecommand \bibitemNoStop [0]{.\EOS\space}%
\providecommand \EOS [0]{\spacefactor3000\relax}%
\providecommand \BibitemShut  [1]{\csname bibitem#1\endcsname}%
\let\auto@bib@innerbib\@empty
\bibitem [{\citenamefont {Giovannetti}\ \emph {et~al.}(2011)\citenamefont {Giovannetti}, \citenamefont {Lloyd},\ and\ \citenamefont {Maccone}}]{giovannettiAdvancesQuantumMetrology2011}%
  \BibitemOpen
  \bibfield  {author} {\bibinfo {author} {\bibfnamefont {V.}~\bibnamefont {Giovannetti}}, \bibinfo {author} {\bibfnamefont {S.}~\bibnamefont {Lloyd}}, \ and\ \bibinfo {author} {\bibfnamefont {L.}~\bibnamefont {Maccone}},\ }\href {\doibase 10.1038/nphoton.2011.35} {\bibfield  {journal} {\bibinfo  {journal} {Nature Photon}\ }\textbf {\bibinfo {volume} {5}},\ \bibinfo {pages} {222} (\bibinfo {year} {2011})}\BibitemShut {NoStop}%
\bibitem [{\citenamefont {Degen}\ \emph {et~al.}(2017)\citenamefont {Degen}, \citenamefont {Reinhard},\ and\ \citenamefont {Cappellaro}}]{degenQuantumSensing2017}%
  \BibitemOpen
  \bibfield  {author} {\bibinfo {author} {\bibfnamefont {C.~L.}\ \bibnamefont {Degen}}, \bibinfo {author} {\bibfnamefont {F.}~\bibnamefont {Reinhard}}, \ and\ \bibinfo {author} {\bibfnamefont {P.}~\bibnamefont {Cappellaro}},\ }\href {\doibase 10.1103/RevModPhys.89.035002} {\bibfield  {journal} {\bibinfo  {journal} {Rev. Mod. Phys.}\ }\textbf {\bibinfo {volume} {89}},\ \bibinfo {pages} {035002} (\bibinfo {year} {2017})}\BibitemShut {NoStop}%
\bibitem [{\citenamefont {Huelga}\ \emph {et~al.}(1997)\citenamefont {Huelga}, \citenamefont {Macchiavello}, \citenamefont {Pellizzari}, \citenamefont {Ekert}, \citenamefont {Plenio},\ and\ \citenamefont {Cirac}}]{huelgaImprovementFrequencyStandardsQuantumEntanglement1997}%
  \BibitemOpen
  \bibfield  {author} {\bibinfo {author} {\bibfnamefont {S.~F.}\ \bibnamefont {Huelga}}, \bibinfo {author} {\bibfnamefont {C.}~\bibnamefont {Macchiavello}}, \bibinfo {author} {\bibfnamefont {T.}~\bibnamefont {Pellizzari}}, \bibinfo {author} {\bibfnamefont {A.~K.}\ \bibnamefont {Ekert}}, \bibinfo {author} {\bibfnamefont {M.~B.}\ \bibnamefont {Plenio}}, \ and\ \bibinfo {author} {\bibfnamefont {J.~I.}\ \bibnamefont {Cirac}},\ }\href {\doibase 10.1103/PhysRevLett.79.3865} {\bibfield  {journal} {\bibinfo  {journal} {Phys. Rev. Lett.}\ }\textbf {\bibinfo {volume} {79}},\ \bibinfo {pages} {3865} (\bibinfo {year} {1997})}\BibitemShut {NoStop}%
\bibitem [{\citenamefont {Dorner}\ \emph {et~al.}(2009)\citenamefont {Dorner}, \citenamefont {{Demkowicz-Dobrzanski}}, \citenamefont {Smith}, \citenamefont {Lundeen}, \citenamefont {Wasilewski}, \citenamefont {Banaszek},\ and\ \citenamefont {Walmsley}}]{dornerOptimalQuantumPhaseEstimation2009}%
  \BibitemOpen
  \bibfield  {author} {\bibinfo {author} {\bibfnamefont {U.}~\bibnamefont {Dorner}}, \bibinfo {author} {\bibfnamefont {R.}~\bibnamefont {{Demkowicz-Dobrzanski}}}, \bibinfo {author} {\bibfnamefont {B.~J.}\ \bibnamefont {Smith}}, \bibinfo {author} {\bibfnamefont {J.~S.}\ \bibnamefont {Lundeen}}, \bibinfo {author} {\bibfnamefont {W.}~\bibnamefont {Wasilewski}}, \bibinfo {author} {\bibfnamefont {K.}~\bibnamefont {Banaszek}}, \ and\ \bibinfo {author} {\bibfnamefont {I.~A.}\ \bibnamefont {Walmsley}},\ }\href {\doibase 10.1103/PhysRevLett.102.040403} {\bibfield  {journal} {\bibinfo  {journal} {Phys. Rev. Lett.}\ }\textbf {\bibinfo {volume} {102}},\ \bibinfo {pages} {040403} (\bibinfo {year} {2009})}\BibitemShut {NoStop}%
\bibitem [{\citenamefont {Escher}\ \emph {et~al.}(2011)\citenamefont {Escher}, \citenamefont {Filho},\ and\ \citenamefont {Davidovich}}]{escherGeneralFrameworkEstimatingUltimatePrecisionLimit2011}%
  \BibitemOpen
  \bibfield  {author} {\bibinfo {author} {\bibfnamefont {B.~M.}\ \bibnamefont {Escher}}, \bibinfo {author} {\bibfnamefont {R.~L. d.~M.}\ \bibnamefont {Filho}}, \ and\ \bibinfo {author} {\bibfnamefont {L.}~\bibnamefont {Davidovich}},\ }\href {\doibase 10.1038/nphys1958} {\bibfield  {journal} {\bibinfo  {journal} {Nat. Phys.}\ }\textbf {\bibinfo {volume} {7}},\ \bibinfo {pages} {406} (\bibinfo {year} {2011})}\BibitemShut {NoStop}%
\bibitem [{\citenamefont {{Demkowicz-Dobrza{\'n}ski}}\ \emph {et~al.}(2012)\citenamefont {{Demkowicz-Dobrza{\'n}ski}}, \citenamefont {Ko{\l}ody{\'n}ski},\ and\ \citenamefont {Gu{\c t}{\u a}}}]{demkowicz-dobrzanskiElusiveHeisenbergLimitQuantumenhancedMetrology2012}%
  \BibitemOpen
  \bibfield  {author} {\bibinfo {author} {\bibfnamefont {R.}~\bibnamefont {{Demkowicz-Dobrza{\'n}ski}}}, \bibinfo {author} {\bibfnamefont {J.}~\bibnamefont {Ko{\l}ody{\'n}ski}}, \ and\ \bibinfo {author} {\bibfnamefont {M.}~\bibnamefont {Gu{\c t}{\u a}}},\ }\href {\doibase 10.1038/ncomms2067} {\bibfield  {journal} {\bibinfo  {journal} {Nat. Commun.}\ }\textbf {\bibinfo {volume} {3}},\ \bibinfo {pages} {1063} (\bibinfo {year} {2012})}\BibitemShut {NoStop}%
\bibitem [{\citenamefont {Ko{\l}ody{\'n}ski}\ and\ \citenamefont {{Demkowicz-Dobrza{\'n}ski}}(2013)}]{kolodynskiEfficientToolsQuantumMetrologyUncorrelatedNoise2013}%
  \BibitemOpen
  \bibfield  {author} {\bibinfo {author} {\bibfnamefont {J.}~\bibnamefont {Ko{\l}ody{\'n}ski}}\ and\ \bibinfo {author} {\bibfnamefont {R.}~\bibnamefont {{Demkowicz-Dobrza{\'n}ski}}},\ }\href {\doibase 10.1088/1367-2630/15/7/073043} {\bibfield  {journal} {\bibinfo  {journal} {New J. Phys.}\ }\textbf {\bibinfo {volume} {15}},\ \bibinfo {pages} {073043} (\bibinfo {year} {2013})}\BibitemShut {NoStop}%
\bibitem [{\citenamefont {{Thomas-Peter}}\ \emph {et~al.}(2011)\citenamefont {{Thomas-Peter}}, \citenamefont {Smith}, \citenamefont {Datta}, \citenamefont {Zhang}, \citenamefont {Dorner},\ and\ \citenamefont {Walmsley}}]{thomas-peterRealworldQuantumSensorsEvaluatingResourcesPrecision2011}%
  \BibitemOpen
  \bibfield  {author} {\bibinfo {author} {\bibfnamefont {N.}~\bibnamefont {{Thomas-Peter}}}, \bibinfo {author} {\bibfnamefont {B.~J.}\ \bibnamefont {Smith}}, \bibinfo {author} {\bibfnamefont {A.}~\bibnamefont {Datta}}, \bibinfo {author} {\bibfnamefont {L.}~\bibnamefont {Zhang}}, \bibinfo {author} {\bibfnamefont {U.}~\bibnamefont {Dorner}}, \ and\ \bibinfo {author} {\bibfnamefont {I.~A.}\ \bibnamefont {Walmsley}},\ }\href {\doibase 10.1103/PhysRevLett.107.113603} {\bibfield  {journal} {\bibinfo  {journal} {Phys. Rev. Lett.}\ }\textbf {\bibinfo {volume} {107}},\ \bibinfo {pages} {113603} (\bibinfo {year} {2011})}\BibitemShut {NoStop}%
\bibitem [{\citenamefont {Kessler}\ \emph {et~al.}(2014)\citenamefont {Kessler}, \citenamefont {Lovchinsky}, \citenamefont {Sushkov},\ and\ \citenamefont {Lukin}}]{kesslerQuantumErrorCorrectionMetrology2014}%
  \BibitemOpen
  \bibfield  {author} {\bibinfo {author} {\bibfnamefont {E.~M.}\ \bibnamefont {Kessler}}, \bibinfo {author} {\bibfnamefont {I.}~\bibnamefont {Lovchinsky}}, \bibinfo {author} {\bibfnamefont {A.~O.}\ \bibnamefont {Sushkov}}, \ and\ \bibinfo {author} {\bibfnamefont {M.~D.}\ \bibnamefont {Lukin}},\ }\href {\doibase 10.1103/PhysRevLett.112.150802} {\bibfield  {journal} {\bibinfo  {journal} {Phys. Rev. Lett.}\ }\textbf {\bibinfo {volume} {112}},\ \bibinfo {pages} {150802} (\bibinfo {year} {2014})}\BibitemShut {NoStop}%
\bibitem [{\citenamefont {D{\"u}r}\ \emph {et~al.}(2014)\citenamefont {D{\"u}r}, \citenamefont {Skotiniotis}, \citenamefont {Fr{\"o}wis},\ and\ \citenamefont {Kraus}}]{durImprovedQuantumMetrologyUsingQuantumError2014}%
  \BibitemOpen
  \bibfield  {author} {\bibinfo {author} {\bibfnamefont {W.}~\bibnamefont {D{\"u}r}}, \bibinfo {author} {\bibfnamefont {M.}~\bibnamefont {Skotiniotis}}, \bibinfo {author} {\bibfnamefont {F.}~\bibnamefont {Fr{\"o}wis}}, \ and\ \bibinfo {author} {\bibfnamefont {B.}~\bibnamefont {Kraus}},\ }\href {\doibase 10.1103/PhysRevLett.112.080801} {\bibfield  {journal} {\bibinfo  {journal} {Phys. Rev. Lett.}\ }\textbf {\bibinfo {volume} {112}},\ \bibinfo {pages} {080801} (\bibinfo {year} {2014})}\BibitemShut {NoStop}%
\bibitem [{\citenamefont {Arrad}\ \emph {et~al.}(2014)\citenamefont {Arrad}, \citenamefont {Vinkler}, \citenamefont {Aharonov},\ and\ \citenamefont {Retzker}}]{arradIncreasingSensingResolutionErrorCorrection2013}%
  \BibitemOpen
  \bibfield  {author} {\bibinfo {author} {\bibfnamefont {G.}~\bibnamefont {Arrad}}, \bibinfo {author} {\bibfnamefont {Y.}~\bibnamefont {Vinkler}}, \bibinfo {author} {\bibfnamefont {D.}~\bibnamefont {Aharonov}}, \ and\ \bibinfo {author} {\bibfnamefont {A.}~\bibnamefont {Retzker}},\ }\href {\doibase 10.1103/PhysRevLett.112.150801} {\bibfield  {journal} {\bibinfo  {journal} {Phys. Rev. Lett.}\ }\textbf {\bibinfo {volume} {112}},\ \bibinfo {pages} {150801} (\bibinfo {year} {2014})}\BibitemShut {NoStop}%
\bibitem [{\citenamefont {{Demkowicz-Dobrza{\'n}ski}}\ \emph {et~al.}(2017)\citenamefont {{Demkowicz-Dobrza{\'n}ski}}, \citenamefont {Czajkowski},\ and\ \citenamefont {Sekatski}}]{demkowicz-dobrzanskiAdaptiveQuantumMetrologyGeneralMarkovianNoise2017}%
  \BibitemOpen
  \bibfield  {author} {\bibinfo {author} {\bibfnamefont {R.}~\bibnamefont {{Demkowicz-Dobrza{\'n}ski}}}, \bibinfo {author} {\bibfnamefont {J.}~\bibnamefont {Czajkowski}}, \ and\ \bibinfo {author} {\bibfnamefont {P.}~\bibnamefont {Sekatski}},\ }\href {\doibase 10.1103/PhysRevX.7.041009} {\bibfield  {journal} {\bibinfo  {journal} {Phys. Rev. X}\ }\textbf {\bibinfo {volume} {7}},\ \bibinfo {pages} {041009} (\bibinfo {year} {2017})}\BibitemShut {NoStop}%
\bibitem [{\citenamefont {Zhou}\ \emph {et~al.}(2018)\citenamefont {Zhou}, \citenamefont {Zhang}, \citenamefont {Preskill},\ and\ \citenamefont {Jiang}}]{zhouAchievingHeisenbergLimitQuantumMetrologyUsing2018}%
  \BibitemOpen
  \bibfield  {author} {\bibinfo {author} {\bibfnamefont {S.}~\bibnamefont {Zhou}}, \bibinfo {author} {\bibfnamefont {M.}~\bibnamefont {Zhang}}, \bibinfo {author} {\bibfnamefont {J.}~\bibnamefont {Preskill}}, \ and\ \bibinfo {author} {\bibfnamefont {L.}~\bibnamefont {Jiang}},\ }\href {\doibase 10.1038/s41467-017-02510-3} {\bibfield  {journal} {\bibinfo  {journal} {Nat. Commun.}\ }\textbf {\bibinfo {volume} {9}},\ \bibinfo {pages} {78} (\bibinfo {year} {2018})}\BibitemShut {NoStop}%
\bibitem [{\citenamefont {Sekatski}\ \emph {et~al.}(2016)\citenamefont {Sekatski}, \citenamefont {Skotiniotis},\ and\ \citenamefont {D{\"u}r}}]{sekatskiDynamicalDecouplingLeadsImprovedScalingNoisy2016}%
  \BibitemOpen
  \bibfield  {author} {\bibinfo {author} {\bibfnamefont {P.}~\bibnamefont {Sekatski}}, \bibinfo {author} {\bibfnamefont {M.}~\bibnamefont {Skotiniotis}}, \ and\ \bibinfo {author} {\bibfnamefont {W.}~\bibnamefont {D{\"u}r}},\ }\href {\doibase 10.1088/1367-2630/18/7/073034} {\bibfield  {journal} {\bibinfo  {journal} {New J. Phys.}\ }\textbf {\bibinfo {volume} {18}},\ \bibinfo {pages} {073034} (\bibinfo {year} {2016})}\BibitemShut {NoStop}%
\bibitem [{\citenamefont {Poggiali}\ \emph {et~al.}(2018)\citenamefont {Poggiali}, \citenamefont {Cappellaro},\ and\ \citenamefont {Fabbri}}]{poggialiOptimalControlOnequbitQuantumSensing2017}%
  \BibitemOpen
  \bibfield  {author} {\bibinfo {author} {\bibfnamefont {F.}~\bibnamefont {Poggiali}}, \bibinfo {author} {\bibfnamefont {P.}~\bibnamefont {Cappellaro}}, \ and\ \bibinfo {author} {\bibfnamefont {N.}~\bibnamefont {Fabbri}},\ }\href {\doibase 10.1103/PhysRevX.8.021059} {\bibfield  {journal} {\bibinfo  {journal} {Phys. Rev. X}\ }\textbf {\bibinfo {volume} {8}},\ \bibinfo {pages} {021059} (\bibinfo {year} {2018})}\BibitemShut {NoStop}%
\bibitem [{\citenamefont {Liu}\ and\ \citenamefont {Yuan}(2017)}]{liuQuantumParameterEstimationOptimalControl2017}%
  \BibitemOpen
  \bibfield  {author} {\bibinfo {author} {\bibfnamefont {J.}~\bibnamefont {Liu}}\ and\ \bibinfo {author} {\bibfnamefont {H.}~\bibnamefont {Yuan}},\ }\href {\doibase 10.1103/PhysRevA.96.012117} {\bibfield  {journal} {\bibinfo  {journal} {Phys. Rev. A}\ }\textbf {\bibinfo {volume} {96}},\ \bibinfo {pages} {012117} (\bibinfo {year} {2017})}\BibitemShut {NoStop}%
\bibitem [{\citenamefont {Marciniak}\ \emph {et~al.}(2022)\citenamefont {Marciniak}, \citenamefont {Feldker}, \citenamefont {Pogorelov}, \citenamefont {Kaubruegger}, \citenamefont {Vasilyev}, \citenamefont {{van Bijnen}}, \citenamefont {Schindler}, \citenamefont {Zoller}, \citenamefont {Blatt},\ and\ \citenamefont {Monz}}]{marciniakOptimalMetrologyProgrammableQuantumSensors2022}%
  \BibitemOpen
  \bibfield  {author} {\bibinfo {author} {\bibfnamefont {C.~D.}\ \bibnamefont {Marciniak}}, \bibinfo {author} {\bibfnamefont {T.}~\bibnamefont {Feldker}}, \bibinfo {author} {\bibfnamefont {I.}~\bibnamefont {Pogorelov}}, \bibinfo {author} {\bibfnamefont {R.}~\bibnamefont {Kaubruegger}}, \bibinfo {author} {\bibfnamefont {D.~V.}\ \bibnamefont {Vasilyev}}, \bibinfo {author} {\bibfnamefont {R.}~\bibnamefont {{van Bijnen}}}, \bibinfo {author} {\bibfnamefont {P.}~\bibnamefont {Schindler}}, \bibinfo {author} {\bibfnamefont {P.}~\bibnamefont {Zoller}}, \bibinfo {author} {\bibfnamefont {R.}~\bibnamefont {Blatt}}, \ and\ \bibinfo {author} {\bibfnamefont {T.}~\bibnamefont {Monz}},\ }\href {\doibase 10.1038/s41586-022-04435-4} {\bibfield  {journal} {\bibinfo  {journal} {Nature}\ }\textbf {\bibinfo {volume} {603}},\ \bibinfo {pages} {604} (\bibinfo {year} {2022})}\BibitemShut {NoStop}%
\bibitem [{\citenamefont {Fallani}\ \emph {et~al.}(2022)\citenamefont {Fallani}, \citenamefont {Rossi}, \citenamefont {Tamascelli},\ and\ \citenamefont {Genoni}}]{fallaniLearningFeedbackControlStrategiesQuantumMetrology2022}%
  \BibitemOpen
  \bibfield  {author} {\bibinfo {author} {\bibfnamefont {A.}~\bibnamefont {Fallani}}, \bibinfo {author} {\bibfnamefont {M.~A.~C.}\ \bibnamefont {Rossi}}, \bibinfo {author} {\bibfnamefont {D.}~\bibnamefont {Tamascelli}}, \ and\ \bibinfo {author} {\bibfnamefont {M.~G.}\ \bibnamefont {Genoni}},\ }\href {\doibase 10.1103/PRXQuantum.3.020310} {\bibfield  {journal} {\bibinfo  {journal} {PRX Quantum}\ }\textbf {\bibinfo {volume} {3}},\ \bibinfo {pages} {020310} (\bibinfo {year} {2022})}\BibitemShut {NoStop}%
\bibitem [{\citenamefont {Zhai}\ \emph {et~al.}(2023)\citenamefont {Zhai}, \citenamefont {Yang}, \citenamefont {Tang}, \citenamefont {Long}, \citenamefont {Nie}, \citenamefont {Xin}, \citenamefont {Lu},\ and\ \citenamefont {Li}}]{zhaiControlenhancedQuantumMetrologyMarkovianNoise2023}%
  \BibitemOpen
  \bibfield  {author} {\bibinfo {author} {\bibfnamefont {Y.}~\bibnamefont {Zhai}}, \bibinfo {author} {\bibfnamefont {X.}~\bibnamefont {Yang}}, \bibinfo {author} {\bibfnamefont {K.}~\bibnamefont {Tang}}, \bibinfo {author} {\bibfnamefont {X.}~\bibnamefont {Long}}, \bibinfo {author} {\bibfnamefont {X.}~\bibnamefont {Nie}}, \bibinfo {author} {\bibfnamefont {T.}~\bibnamefont {Xin}}, \bibinfo {author} {\bibfnamefont {D.}~\bibnamefont {Lu}}, \ and\ \bibinfo {author} {\bibfnamefont {J.}~\bibnamefont {Li}},\ }\href {\doibase 10.1103/PhysRevA.107.022602} {\bibfield  {journal} {\bibinfo  {journal} {Phys. Rev. A}\ }\textbf {\bibinfo {volume} {107}},\ \bibinfo {pages} {022602} (\bibinfo {year} {2023})}\BibitemShut {NoStop}%
\bibitem [{\citenamefont {Yang}\ \emph {et~al.}(2024)\citenamefont {Yang}, \citenamefont {Long}, \citenamefont {Liu}, \citenamefont {Tang}, \citenamefont {Zhai}, \citenamefont {Nie}, \citenamefont {Xin}, \citenamefont {Li},\ and\ \citenamefont {Lu}}]{yangControlenhancedNonMarkovianQuantumMetrology2024}%
  \BibitemOpen
  \bibfield  {author} {\bibinfo {author} {\bibfnamefont {X.}~\bibnamefont {Yang}}, \bibinfo {author} {\bibfnamefont {X.}~\bibnamefont {Long}}, \bibinfo {author} {\bibfnamefont {R.}~\bibnamefont {Liu}}, \bibinfo {author} {\bibfnamefont {K.}~\bibnamefont {Tang}}, \bibinfo {author} {\bibfnamefont {Y.}~\bibnamefont {Zhai}}, \bibinfo {author} {\bibfnamefont {X.}~\bibnamefont {Nie}}, \bibinfo {author} {\bibfnamefont {T.}~\bibnamefont {Xin}}, \bibinfo {author} {\bibfnamefont {J.}~\bibnamefont {Li}}, \ and\ \bibinfo {author} {\bibfnamefont {D.}~\bibnamefont {Lu}},\ }\href {\doibase 10.1038/s42005-024-01758-8} {\bibfield  {journal} {\bibinfo  {journal} {Commun. Phys.}\ }\textbf {\bibinfo {volume} {7}},\ \bibinfo {pages} {1} (\bibinfo {year} {2024})}\BibitemShut {NoStop}%
\bibitem [{\citenamefont {Bai}\ and\ \citenamefont {An}(2023)}]{baiFloquetEngineeringOvercomeNoGoTheoremNoisy2023}%
  \BibitemOpen
  \bibfield  {author} {\bibinfo {author} {\bibfnamefont {S.-Y.}\ \bibnamefont {Bai}}\ and\ \bibinfo {author} {\bibfnamefont {J.-H.}\ \bibnamefont {An}},\ }\href {\doibase 10.1103/PhysRevLett.131.050801} {\bibfield  {journal} {\bibinfo  {journal} {Phys. Rev. Lett.}\ }\textbf {\bibinfo {volume} {131}},\ \bibinfo {pages} {050801} (\bibinfo {year} {2023})}\BibitemShut {NoStop}%
\bibitem [{\citenamefont {Rojkov}\ \emph {et~al.}(2022)\citenamefont {Rojkov}, \citenamefont {Layden}, \citenamefont {Cappellaro}, \citenamefont {Home},\ and\ \citenamefont {Reiter}}]{rojkovBiasErrorcorrectedQuantumSensing2022}%
  \BibitemOpen
  \bibfield  {author} {\bibinfo {author} {\bibfnamefont {I.}~\bibnamefont {Rojkov}}, \bibinfo {author} {\bibfnamefont {D.}~\bibnamefont {Layden}}, \bibinfo {author} {\bibfnamefont {P.}~\bibnamefont {Cappellaro}}, \bibinfo {author} {\bibfnamefont {J.}~\bibnamefont {Home}}, \ and\ \bibinfo {author} {\bibfnamefont {F.}~\bibnamefont {Reiter}},\ }\href {\doibase 10.1103/PhysRevLett.128.140503} {\bibfield  {journal} {\bibinfo  {journal} {Phys. Rev. Lett.}\ }\textbf {\bibinfo {volume} {128}},\ \bibinfo {pages} {140503} (\bibinfo {year} {2022})}\BibitemShut {NoStop}%
\bibitem [{\citenamefont {Mihailescu}\ \emph {et~al.}(2024)\citenamefont {Mihailescu}, \citenamefont {Campbell},\ and\ \citenamefont {Gietka}}]{mihailescuUncertainQuantumCriticalMetrologySingleMulti2024a}%
  \BibitemOpen
  \bibfield  {author} {\bibinfo {author} {\bibfnamefont {G.}~\bibnamefont {Mihailescu}}, \bibinfo {author} {\bibfnamefont {S.}~\bibnamefont {Campbell}}, \ and\ \bibinfo {author} {\bibfnamefont {K.}~\bibnamefont {Gietka}},\ }\href {https://arxiv.org/abs/2407.19917} {\enquote {\bibinfo {title} {Uncertain {{Quantum Critical Metrology}}: {{From Single}} to {{Multi Parameter Sensing}}},}\ } (\bibinfo {year} {2024}),\ \Eprint {http://arxiv.org/abs/2407.19917} {arXiv:2407.19917 [quant-ph]} \BibitemShut {NoStop}%
\bibitem [{\citenamefont {Peres}(1980)}]{peres1980quantum-clocks}%
  \BibitemOpen
  \bibfield  {author} {\bibinfo {author} {\bibfnamefont {A.}~\bibnamefont {Peres}},\ }\href {\doibase 10.1119/1.12061} {\bibfield  {journal} {\bibinfo  {journal} {Am. J. Phys.}\ }\textbf {\bibinfo {volume} {48}},\ \bibinfo {pages} {552} (\bibinfo {year} {1980})}\BibitemShut {NoStop}%
\bibitem [{\citenamefont {Xuereb}\ \emph {et~al.}(2023)\citenamefont {Xuereb}, \citenamefont {Meier}, \citenamefont {Erker}, \citenamefont {Mitchison},\ and\ \citenamefont {Huber}}]{xuerebImpactImperfectTimekeepingQuantumControl2023}%
  \BibitemOpen
  \bibfield  {author} {\bibinfo {author} {\bibfnamefont {J.}~\bibnamefont {Xuereb}}, \bibinfo {author} {\bibfnamefont {F.}~\bibnamefont {Meier}}, \bibinfo {author} {\bibfnamefont {P.}~\bibnamefont {Erker}}, \bibinfo {author} {\bibfnamefont {M.~T.}\ \bibnamefont {Mitchison}}, \ and\ \bibinfo {author} {\bibfnamefont {M.}~\bibnamefont {Huber}},\ }\href {\doibase https://doi.org/10.1103/PhysRevLett.131.160204} {\bibfield  {journal} {\bibinfo  {journal} {Phys. Rev. Lett.}\ }\textbf {\bibinfo {volume} {131}},\ \bibinfo {pages} {160204} (\bibinfo {year} {2023})}\BibitemShut {NoStop}%
\bibitem [{\citenamefont {Ball}\ \emph {et~al.}(2016)\citenamefont {Ball}, \citenamefont {Oliver},\ and\ \citenamefont {Biercuk}}]{ballRoleMasterClockStabilityQuantumInformation2016}%
  \BibitemOpen
  \bibfield  {author} {\bibinfo {author} {\bibfnamefont {H.}~\bibnamefont {Ball}}, \bibinfo {author} {\bibfnamefont {W.~D.}\ \bibnamefont {Oliver}}, \ and\ \bibinfo {author} {\bibfnamefont {M.~J.}\ \bibnamefont {Biercuk}},\ }\href {\doibase 10.1038/npjqi.2016.33} {\bibfield  {journal} {\bibinfo  {journal} {npj Quantum Inf}\ }\textbf {\bibinfo {volume} {2}},\ \bibinfo {pages} {1} (\bibinfo {year} {2016})}\BibitemShut {NoStop}%
\bibitem [{\citenamefont {Penrose}(1979)}]{penroseFoundationsStatisticalMechanics1979}%
  \BibitemOpen
  \bibfield  {author} {\bibinfo {author} {\bibfnamefont {O.}~\bibnamefont {Penrose}},\ }\href {\doibase 10.1088/0034-4885/42/12/002} {\bibfield  {journal} {\bibinfo  {journal} {Rep. Prog. Phys.}\ }\textbf {\bibinfo {volume} {42}},\ \bibinfo {pages} {1937} (\bibinfo {year} {1979})}\BibitemShut {NoStop}%
\bibitem [{\citenamefont {Deutsch}(2018)}]{deutschEigenstateThermalizationHypothesis2018}%
  \BibitemOpen
  \bibfield  {author} {\bibinfo {author} {\bibfnamefont {J.~M.}\ \bibnamefont {Deutsch}},\ }\href {\doibase 10.1088/1361-6633/aac9f1} {\bibfield  {journal} {\bibinfo  {journal} {Rep. Prog. Phys.}\ }\textbf {\bibinfo {volume} {81}},\ \bibinfo {pages} {082001} (\bibinfo {year} {2018})}\BibitemShut {NoStop}%
\bibitem [{\citenamefont {Srednicki}(1994)}]{srednickiChaosQuantumThermalization1994}%
  \BibitemOpen
  \bibfield  {author} {\bibinfo {author} {\bibfnamefont {M.}~\bibnamefont {Srednicki}},\ }\href {\doibase 10.1103/PhysRevE.50.888} {\bibfield  {journal} {\bibinfo  {journal} {Phys. Rev. E}\ }\textbf {\bibinfo {volume} {50}},\ \bibinfo {pages} {888} (\bibinfo {year} {1994})}\BibitemShut {NoStop}%
\bibitem [{\citenamefont {Clos}\ \emph {et~al.}(2016)\citenamefont {Clos}, \citenamefont {Porras}, \citenamefont {Warring},\ and\ \citenamefont {Schaetz}}]{closTimeresolvedObservationThermalizationIsolatedQuantumSystem2016}%
  \BibitemOpen
  \bibfield  {author} {\bibinfo {author} {\bibfnamefont {G.}~\bibnamefont {Clos}}, \bibinfo {author} {\bibfnamefont {D.}~\bibnamefont {Porras}}, \bibinfo {author} {\bibfnamefont {U.}~\bibnamefont {Warring}}, \ and\ \bibinfo {author} {\bibfnamefont {T.}~\bibnamefont {Schaetz}},\ }\href {\doibase 10.1103/PhysRevLett.117.170401} {\bibfield  {journal} {\bibinfo  {journal} {Phys. Rev. Lett.}\ }\textbf {\bibinfo {volume} {117}},\ \bibinfo {pages} {170401} (\bibinfo {year} {2016})}\BibitemShut {NoStop}%
\bibitem [{con()}]{constant-time-error}%
  \BibitemOpen
  \href@noop {} {}\bibinfo {note} {The form of error differs from the most general time domain imperfection by a constant error that is independent of $t$, which can be solved through existing methods.}\BibitemShut {Stop}%
\bibitem [{\citenamefont {Van~Tree}(2001)}]{treeDetectionEstimationModulationTheory2001}%
  \BibitemOpen
  \bibfield  {author} {\bibinfo {author} {\bibfnamefont {H.~L.}\ \bibnamefont {Van~Tree}},\ }\href {\doibase 10.1002/0471221082} {\emph {\bibinfo {title} {Detection, {{Estimation}}, and {{Modulation Theory}}}}},\ \bibinfo {edition} {1st}\ ed.\ (\bibinfo  {publisher} {John Wiley \& Sons, Ltd},\ \bibinfo {year} {2001})\BibitemShut {NoStop}%
\bibitem [{\citenamefont {Braunstein}\ and\ \citenamefont {Caves}(1994)}]{braunsteinStatisticalDistanceGeometryQuantumStates1994}%
  \BibitemOpen
  \bibfield  {author} {\bibinfo {author} {\bibfnamefont {S.~L.}\ \bibnamefont {Braunstein}}\ and\ \bibinfo {author} {\bibfnamefont {C.~M.}\ \bibnamefont {Caves}},\ }\href {\doibase 10.1103/PhysRevLett.72.3439} {\bibfield  {journal} {\bibinfo  {journal} {Phys. Rev. Lett.}\ }\textbf {\bibinfo {volume} {72}},\ \bibinfo {pages} {3439} (\bibinfo {year} {1994})}\BibitemShut {NoStop}%
\bibitem [{\citenamefont {Giovannetti}\ \emph {et~al.}(2006)\citenamefont {Giovannetti}, \citenamefont {Lloyd},\ and\ \citenamefont {Maccone}}]{giovannettiQuantumMetrology2006}%
  \BibitemOpen
  \bibfield  {author} {\bibinfo {author} {\bibfnamefont {V.}~\bibnamefont {Giovannetti}}, \bibinfo {author} {\bibfnamefont {S.}~\bibnamefont {Lloyd}}, \ and\ \bibinfo {author} {\bibfnamefont {L.}~\bibnamefont {Maccone}},\ }\href {\doibase 10.1103/PhysRevLett.96.010401} {\bibfield  {journal} {\bibinfo  {journal} {Phys. Rev. Lett.}\ }\textbf {\bibinfo {volume} {96}},\ \bibinfo {pages} {010401} (\bibinfo {year} {2006})}\BibitemShut {NoStop}%
\bibitem [{SM()}]{SM}%
  \BibitemOpen
  \href@noop {} {}\bibinfo {note} {See Supplementary Material for detailed calculations in FE case, proof of Theorem 1, concrete measurement protocols, and numerical results for multilevel environments.}\BibitemShut {Stop}%
\bibitem [{exp({\natexlab{a}})}]{explain-assumption-opt-bias}%
  \BibitemOpen
  \href@noop {} {}\bibinfo {note} {We note that the estimator is unbiased when $\omega$ equals the reference value $0$, which is utilized to relax the bound. However, to prevent numerical calculation errors, we avoid setting $\omega$ precisely to $0$. For more details, please refer to the Supplemental Material.}\BibitemShut {Stop}%
\bibitem [{cor()}]{corrections-loss-FE}%
  \BibitemOpen
  \href@noop {} {}\bibinfo {note} {For conciseness, we present the higher-order terms here: $O\left[{\nu\omega^8 T^2}/{(g^2+\nu T^2\omega^4)^2}\right]$ for $\epsilon \ge \epsilon^\ast$ and $O\bm{(}{\omega ^4}/\{{\nu T^2 [g^2 \cos ^2(T \Omega \epsilon )+\omega ^2]^2}\}\bm{)}$ for $\epsilon < \epsilon^\ast$.}\BibitemShut {Stop}%
\bibitem [{exp({\natexlab{b}})}]{explain-fast-ctrl}%
  \BibitemOpen
  \href@noop {} {}\bibinfo {note} {For conciseness, we assume that the control pulse strength significantly exceeds that of the Hamiltonian $H$, approximating an instantaneous control pulse. This is nevertheless influenced by the TDI.}\BibitemShut {Stop}%
\bibitem [{exp({\natexlab{c}})}]{explain-epsilon0-optimal}%
  \BibitemOpen
  \href@noop {} {}\bibinfo {note} {The ``optimality" refers to strategies that are optimal when TDI is absent. While it is nearly impossible to determine optimal strategies for general TDI, we anticipate that no strategy can significantly outperform the one we have employed here. For a detailed description of the protocols and the proof of the optimality, please refer to the Supplemental Material.}\BibitemShut {Stop}%
\bibitem [{\citenamefont {Lidar}\ \emph {et~al.}(1998)\citenamefont {Lidar}, \citenamefont {Chuang},\ and\ \citenamefont {Whaley}}]{lidarDecoherenceFreeSubspacesQuantumComputation1998}%
  \BibitemOpen
  \bibfield  {author} {\bibinfo {author} {\bibfnamefont {D.~A.}\ \bibnamefont {Lidar}}, \bibinfo {author} {\bibfnamefont {I.~L.}\ \bibnamefont {Chuang}}, \ and\ \bibinfo {author} {\bibfnamefont {K.~B.}\ \bibnamefont {Whaley}},\ }\href {\doibase 10.1103/PhysRevLett.81.2594} {\bibfield  {journal} {\bibinfo  {journal} {Phys. Rev. Lett.}\ }\textbf {\bibinfo {volume} {81}},\ \bibinfo {pages} {2594} (\bibinfo {year} {1998})}\BibitemShut {NoStop}%
\bibitem [{\citenamefont {Facchi}\ and\ \citenamefont {Pascazio}(2008)}]{facchiQuantumZenoDynamicsMathematicalPhysicalAspects2008}%
  \BibitemOpen
  \bibfield  {author} {\bibinfo {author} {\bibfnamefont {P.}~\bibnamefont {Facchi}}\ and\ \bibinfo {author} {\bibfnamefont {S.}~\bibnamefont {Pascazio}},\ }\href {\doibase 10.1088/1751-8113/41/49/493001} {\bibfield  {journal} {\bibinfo  {journal} {J. Phys. A: Math. Theor.}\ }\textbf {\bibinfo {volume} {41}},\ \bibinfo {pages} {493001} (\bibinfo {year} {2008})}\BibitemShut {NoStop}%
\bibitem [{\citenamefont {Burgarth}\ \emph {et~al.}(2019)\citenamefont {Burgarth}, \citenamefont {Facchi}, \citenamefont {Nakazato}, \citenamefont {Pascazio},\ and\ \citenamefont {Yuasa}}]{burgarthGeneralizedAdiabaticTheoremStrongCouplingLimits2019}%
  \BibitemOpen
  \bibfield  {author} {\bibinfo {author} {\bibfnamefont {D.}~\bibnamefont {Burgarth}}, \bibinfo {author} {\bibfnamefont {P.}~\bibnamefont {Facchi}}, \bibinfo {author} {\bibfnamefont {H.}~\bibnamefont {Nakazato}}, \bibinfo {author} {\bibfnamefont {S.}~\bibnamefont {Pascazio}}, \ and\ \bibinfo {author} {\bibfnamefont {K.}~\bibnamefont {Yuasa}},\ }\href {\doibase 10.22331/q-2019-06-12-152} {\bibfield  {journal} {\bibinfo  {journal} {Quantum}\ }\textbf {\bibinfo {volume} {3}},\ \bibinfo {pages} {152} (\bibinfo {year} {2019})}\BibitemShut {NoStop}%
\bibitem [{cla()}]{classical-interpretation-SWAP}%
  \BibitemOpen
  \href@noop {} {}\bibinfo {note} {Note that the multiplier $\omega/2$ in the system Hamiltonian is ignored for conciseness. In addition, we further remark that this Zeno dynamcis has a classical analogue that aids in its understanding. The SWAP interaction can be interpreted as a continuous exchange of particles between the system and the environment. As a result, the evolution of the system particle proceeds at half the original speed. This loss can be compensated for by preparing the state in the DES.}\BibitemShut {Stop}%
\bibitem [{\citenamefont {Long}\ \emph {et~al.}(2022)\citenamefont {Long}, \citenamefont {He}, \citenamefont {Zhang}, \citenamefont {Tang}, \citenamefont {Lin}, \citenamefont {Liu}, \citenamefont {Nie}, \citenamefont {Feng}, \citenamefont {Li}, \citenamefont {Xin}, \citenamefont {Ai},\ and\ \citenamefont {Lu}}]{longEntanglementEnhancedQuantumMetrologyColoredNoiseQuantum2022}%
  \BibitemOpen
  \bibfield  {author} {\bibinfo {author} {\bibfnamefont {X.}~\bibnamefont {Long}}, \bibinfo {author} {\bibfnamefont {W.-T.}\ \bibnamefont {He}}, \bibinfo {author} {\bibfnamefont {N.-N.}\ \bibnamefont {Zhang}}, \bibinfo {author} {\bibfnamefont {K.}~\bibnamefont {Tang}}, \bibinfo {author} {\bibfnamefont {Z.}~\bibnamefont {Lin}}, \bibinfo {author} {\bibfnamefont {H.}~\bibnamefont {Liu}}, \bibinfo {author} {\bibfnamefont {X.}~\bibnamefont {Nie}}, \bibinfo {author} {\bibfnamefont {G.}~\bibnamefont {Feng}}, \bibinfo {author} {\bibfnamefont {J.}~\bibnamefont {Li}}, \bibinfo {author} {\bibfnamefont {T.}~\bibnamefont {Xin}}, \bibinfo {author} {\bibfnamefont {Q.}~\bibnamefont {Ai}}, \ and\ \bibinfo {author} {\bibfnamefont {D.}~\bibnamefont {Lu}},\ }\href {\doibase 10.1103/PhysRevLett.129.070502} {\bibfield  {journal} {\bibinfo  {journal} {Phys. Rev. Lett.}\ }\textbf {\bibinfo {volume} {129}},\ \bibinfo {pages} {070502} (\bibinfo {year} {2022})}\BibitemShut {NoStop}%
\bibitem [{exp({\natexlab{d}})}]{explain-interaction-channel}%
  \BibitemOpen
  \href@noop {} {}\bibinfo {note} {These two interactions will be reduced to amplitude damping and dephasing channels, respectively, if the environment is memoryless.}\BibitemShut {Stop}%
\bibitem [{rem()}]{remarks-larger-error-ce}%
  \BibitemOpen
  \href@noop {} {}\bibinfo {note} {Note that in some instances of the CE case, the experimental error exceeds the theoretical upper bound. This can be attributed to imperfections of the experimental devices beyond TDI, such as the electromagnetic noise from electronic components, radiofrequency coils, and the associated circuits, which are not considered in the theoretical analysis.}\BibitemShut {Stop}%
\bibitem [{\citenamefont {Bernien}\ \emph {et~al.}(2017)\citenamefont {Bernien}, \citenamefont {Schwartz}, \citenamefont {Keesling}, \citenamefont {Levine}, \citenamefont {Omran}, \citenamefont {Pichler}, \citenamefont {Choi}, \citenamefont {Zibrov}, \citenamefont {Endres}, \citenamefont {Greiner}, \citenamefont {Vuleti{\'c}},\ and\ \citenamefont {Lukin}}]{bernienProbingManybodyDynamics51atomQuantumSimulator2017}%
  \BibitemOpen
  \bibfield  {author} {\bibinfo {author} {\bibfnamefont {H.}~\bibnamefont {Bernien}}, \bibinfo {author} {\bibfnamefont {S.}~\bibnamefont {Schwartz}}, \bibinfo {author} {\bibfnamefont {A.}~\bibnamefont {Keesling}}, \bibinfo {author} {\bibfnamefont {H.}~\bibnamefont {Levine}}, \bibinfo {author} {\bibfnamefont {A.}~\bibnamefont {Omran}}, \bibinfo {author} {\bibfnamefont {H.}~\bibnamefont {Pichler}}, \bibinfo {author} {\bibfnamefont {S.}~\bibnamefont {Choi}}, \bibinfo {author} {\bibfnamefont {A.~S.}\ \bibnamefont {Zibrov}}, \bibinfo {author} {\bibfnamefont {M.}~\bibnamefont {Endres}}, \bibinfo {author} {\bibfnamefont {M.}~\bibnamefont {Greiner}}, \bibinfo {author} {\bibfnamefont {V.}~\bibnamefont {Vuleti{\'c}}}, \ and\ \bibinfo {author} {\bibfnamefont {M.~D.}\ \bibnamefont {Lukin}},\ }\href {\doibase 10.1038/nature24622} {\bibfield  {journal} {\bibinfo  {journal} {Nature}\ }\textbf {\bibinfo {volume} {551}},\ \bibinfo {pages} {579} (\bibinfo {year} {2017})}\BibitemShut {NoStop}%
\bibitem [{\citenamefont {Serbyn}\ \emph {et~al.}(2021)\citenamefont {Serbyn}, \citenamefont {Abanin},\ and\ \citenamefont {Papi{\'c}}}]{serbynQuantumManyBodyScarsWeakBreakingErgodicity2021}%
  \BibitemOpen
  \bibfield  {author} {\bibinfo {author} {\bibfnamefont {M.}~\bibnamefont {Serbyn}}, \bibinfo {author} {\bibfnamefont {D.~A.}\ \bibnamefont {Abanin}}, \ and\ \bibinfo {author} {\bibfnamefont {Z.}~\bibnamefont {Papi{\'c}}},\ }\href {\doibase 10.1038/s41567-021-01230-2} {\bibfield  {journal} {\bibinfo  {journal} {Nat. Phys.}\ }\textbf {\bibinfo {volume} {17}},\ \bibinfo {pages} {675} (\bibinfo {year} {2021})}\BibitemShut {NoStop}%
\bibitem [{\citenamefont {Malabarba}\ \emph {et~al.}(2015)\citenamefont {Malabarba}, \citenamefont {Short},\ and\ \citenamefont {Kammerlander}}]{malabarbaClockdrivenQuantumThermalEngines2015}%
  \BibitemOpen
  \bibfield  {author} {\bibinfo {author} {\bibfnamefont {A.~S.~L.}\ \bibnamefont {Malabarba}}, \bibinfo {author} {\bibfnamefont {A.~J.}\ \bibnamefont {Short}}, \ and\ \bibinfo {author} {\bibfnamefont {P.}~\bibnamefont {Kammerlander}},\ }\href {\doibase 10.1088/1367-2630/17/4/045027} {\bibfield  {journal} {\bibinfo  {journal} {New J. Phys.}\ }\textbf {\bibinfo {volume} {17}},\ \bibinfo {pages} {045027} (\bibinfo {year} {2015})}\BibitemShut {NoStop}%
\bibitem [{\citenamefont {Erker}\ \emph {et~al.}(2017)\citenamefont {Erker}, \citenamefont {Mitchison}, \citenamefont {Silva}, \citenamefont {Woods}, \citenamefont {Brunner},\ and\ \citenamefont {Huber}}]{erkerAutonomousQuantumClocksDoesThermodynamicsLimit2017}%
  \BibitemOpen
  \bibfield  {author} {\bibinfo {author} {\bibfnamefont {P.}~\bibnamefont {Erker}}, \bibinfo {author} {\bibfnamefont {M.~T.}\ \bibnamefont {Mitchison}}, \bibinfo {author} {\bibfnamefont {R.}~\bibnamefont {Silva}}, \bibinfo {author} {\bibfnamefont {M.~P.}\ \bibnamefont {Woods}}, \bibinfo {author} {\bibfnamefont {N.}~\bibnamefont {Brunner}}, \ and\ \bibinfo {author} {\bibfnamefont {M.}~\bibnamefont {Huber}},\ }\href {\doibase 10.1103/PhysRevX.7.031022} {\bibfield  {journal} {\bibinfo  {journal} {Phys. Rev. X}\ }\textbf {\bibinfo {volume} {7}},\ \bibinfo {pages} {031022} (\bibinfo {year} {2017})}\BibitemShut {NoStop}%
\bibitem [{\citenamefont {Woods}\ \emph {et~al.}(2019)\citenamefont {Woods}, \citenamefont {Silva},\ and\ \citenamefont {Oppenheim}}]{woodsAutonomousQuantumMachinesFiniteSizedClocks2019}%
  \BibitemOpen
  \bibfield  {author} {\bibinfo {author} {\bibfnamefont {M.~P.}\ \bibnamefont {Woods}}, \bibinfo {author} {\bibfnamefont {R.}~\bibnamefont {Silva}}, \ and\ \bibinfo {author} {\bibfnamefont {J.}~\bibnamefont {Oppenheim}},\ }\href {\doibase 10.1007/s00023-018-0736-9} {\bibfield  {journal} {\bibinfo  {journal} {Ann. Henri Poincare}\ }\textbf {\bibinfo {volume} {20}},\ \bibinfo {pages} {125} (\bibinfo {year} {2019})}\BibitemShut {NoStop}%
\bibitem [{\citenamefont {Woods}\ \emph {et~al.}(2022)\citenamefont {Woods}, \citenamefont {Silva}, \citenamefont {P{\"u}tz}, \citenamefont {Stupar},\ and\ \citenamefont {Renner}}]{woodsQuantumClocksAreMoreAccurateClassical2022}%
  \BibitemOpen
  \bibfield  {author} {\bibinfo {author} {\bibfnamefont {M.~P.}\ \bibnamefont {Woods}}, \bibinfo {author} {\bibfnamefont {R.}~\bibnamefont {Silva}}, \bibinfo {author} {\bibfnamefont {G.}~\bibnamefont {P{\"u}tz}}, \bibinfo {author} {\bibfnamefont {S.}~\bibnamefont {Stupar}}, \ and\ \bibinfo {author} {\bibfnamefont {R.}~\bibnamefont {Renner}},\ }\href {\doibase 10.1103/PRXQuantum.3.010319} {\bibfield  {journal} {\bibinfo  {journal} {PRX Quantum}\ }\textbf {\bibinfo {volume} {3}},\ \bibinfo {pages} {010319} (\bibinfo {year} {2022})}\BibitemShut {NoStop}%
\bibitem [{\citenamefont {Woods}\ and\ \citenamefont {Horodecki}(2023)}]{woodsAutonomousQuantumDevicesWhenAreThey2023}%
  \BibitemOpen
  \bibfield  {author} {\bibinfo {author} {\bibfnamefont {M.~P.}\ \bibnamefont {Woods}}\ and\ \bibinfo {author} {\bibfnamefont {M.}~\bibnamefont {Horodecki}},\ }\href {\doibase 10.1103/PhysRevX.13.011016} {\bibfield  {journal} {\bibinfo  {journal} {Physical Review X}\ }\textbf {\bibinfo {volume} {13}},\ \bibinfo {pages} {011016} (\bibinfo {year} {2023})}\BibitemShut {NoStop}%
\bibitem [{\citenamefont {Meier}\ \emph {et~al.}(2023)\citenamefont {Meier}, \citenamefont {Schwarzhans}, \citenamefont {Erker},\ and\ \citenamefont {Huber}}]{meierFundamentalAccuracyResolutionTradeTimekeepingDevices2023}%
  \BibitemOpen
  \bibfield  {author} {\bibinfo {author} {\bibfnamefont {F.}~\bibnamefont {Meier}}, \bibinfo {author} {\bibfnamefont {E.}~\bibnamefont {Schwarzhans}}, \bibinfo {author} {\bibfnamefont {P.}~\bibnamefont {Erker}}, \ and\ \bibinfo {author} {\bibfnamefont {M.}~\bibnamefont {Huber}},\ }\href {\doibase 10.1103/PhysRevLett.131.220201} {\bibfield  {journal} {\bibinfo  {journal} {Phys. Rev. Lett.}\ }\textbf {\bibinfo {volume} {131}},\ \bibinfo {pages} {220201} (\bibinfo {year} {2023})}\BibitemShut {NoStop}%
\bibitem [{\citenamefont {Yang}\ and\ \citenamefont {Renner}(2020)}]{yangUltimateLimitTimeSignalGeneration2020}%
  \BibitemOpen
  \bibfield  {author} {\bibinfo {author} {\bibfnamefont {Y.}~\bibnamefont {Yang}}\ and\ \bibinfo {author} {\bibfnamefont {R.}~\bibnamefont {Renner}},\ }\href {https://arxiv.org/abs/2004.07857} {\enquote {\bibinfo {title} {Ultimate limit on time signal generation},}\ } (\bibinfo {year} {2020}),\ \Eprint {http://arxiv.org/abs/2004.07857} {arXiv:2004.07857 [quant-ph]} \BibitemShut {NoStop}%
\bibitem [{\citenamefont {Chiribella}\ \emph {et~al.}(2008{\natexlab{a}})\citenamefont {Chiribella}, \citenamefont {D'Ariano},\ and\ \citenamefont {Perinotti}}]{chiribellaQuantumCircuitsArchitecture2008}%
  \BibitemOpen
  \bibfield  {author} {\bibinfo {author} {\bibfnamefont {G.}~\bibnamefont {Chiribella}}, \bibinfo {author} {\bibfnamefont {G.~M.}\ \bibnamefont {D'Ariano}}, \ and\ \bibinfo {author} {\bibfnamefont {P.}~\bibnamefont {Perinotti}},\ }\href {\doibase 10.1103/PhysRevLett.101.060401} {\bibfield  {journal} {\bibinfo  {journal} {Phys. Rev. Lett.}\ }\textbf {\bibinfo {volume} {101}},\ \bibinfo {pages} {060401} (\bibinfo {year} {2008}{\natexlab{a}})}\BibitemShut {NoStop}%
\bibitem [{\citenamefont {Chiribella}\ \emph {et~al.}(2008{\natexlab{b}})\citenamefont {Chiribella}, \citenamefont {D'Ariano},\ and\ \citenamefont {Perinotti}}]{chiribellaTransformingQuantumOperationsQuantumSupermaps2008}%
  \BibitemOpen
  \bibfield  {author} {\bibinfo {author} {\bibfnamefont {G.}~\bibnamefont {Chiribella}}, \bibinfo {author} {\bibfnamefont {G.~M.}\ \bibnamefont {D'Ariano}}, \ and\ \bibinfo {author} {\bibfnamefont {P.}~\bibnamefont {Perinotti}},\ }\href {\doibase 10.1209/0295-5075/83/30004} {\bibfield  {journal} {\bibinfo  {journal} {Europhys. Lett.}\ }\textbf {\bibinfo {volume} {83}},\ \bibinfo {pages} {30004} (\bibinfo {year} {2008}{\natexlab{b}})}\BibitemShut {NoStop}%
\bibitem [{\citenamefont {Chiribella}\ \emph {et~al.}(2009)\citenamefont {Chiribella}, \citenamefont {D'Ariano},\ and\ \citenamefont {Perinotti}}]{chiribellaTheoreticalFrameworkQuantumNetworks2009}%
  \BibitemOpen
  \bibfield  {author} {\bibinfo {author} {\bibfnamefont {G.}~\bibnamefont {Chiribella}}, \bibinfo {author} {\bibfnamefont {G.~M.}\ \bibnamefont {D'Ariano}}, \ and\ \bibinfo {author} {\bibfnamefont {P.}~\bibnamefont {Perinotti}},\ }\href {\doibase 10.1103/PhysRevA.80.022339} {\bibfield  {journal} {\bibinfo  {journal} {Phys. Rev. A}\ }\textbf {\bibinfo {volume} {80}},\ \bibinfo {pages} {022339} (\bibinfo {year} {2009})}\BibitemShut {NoStop}%
\bibitem [{\citenamefont {Chiribella}(2012)}]{chiribellaOptimalNetworksQuantumMetrologySemidefinitePrograms2012}%
  \BibitemOpen
  \bibfield  {author} {\bibinfo {author} {\bibfnamefont {G.}~\bibnamefont {Chiribella}},\ }\href {\doibase 10.1088/1367-2630/14/12/125008} {\bibfield  {journal} {\bibinfo  {journal} {New Journal of Physics}\ }\textbf {\bibinfo {volume} {14}},\ \bibinfo {pages} {125008} (\bibinfo {year} {2012})}\BibitemShut {NoStop}%
\bibitem [{\citenamefont {Yang}(2019)}]{yangMemoryEffectsQuantumMetrology2019}%
  \BibitemOpen
  \bibfield  {author} {\bibinfo {author} {\bibfnamefont {Y.}~\bibnamefont {Yang}},\ }\href {\doibase 10.1103/PhysRevLett.123.110501} {\bibfield  {journal} {\bibinfo  {journal} {Phys. Rev. Lett.}\ }\textbf {\bibinfo {volume} {123}},\ \bibinfo {pages} {110501} (\bibinfo {year} {2019})}\BibitemShut {NoStop}%
\bibitem [{\citenamefont {Altherr}\ and\ \citenamefont {Yang}(2021)}]{altherrQuantumMetrologyNonMarkovianProcesses2021}%
  \BibitemOpen
  \bibfield  {author} {\bibinfo {author} {\bibfnamefont {A.}~\bibnamefont {Altherr}}\ and\ \bibinfo {author} {\bibfnamefont {Y.}~\bibnamefont {Yang}},\ }\href {\doibase 10.1103/PhysRevLett.127.060501} {\bibfield  {journal} {\bibinfo  {journal} {Phys. Rev. Lett.}\ }\textbf {\bibinfo {volume} {127}},\ \bibinfo {pages} {060501} (\bibinfo {year} {2021})}\BibitemShut {NoStop}%
\bibitem [{\citenamefont {Kay}\ and\ \citenamefont {Eldar}(2008)}]{kayRethinkingBiasedEstimationLectureNotes2008}%
  \BibitemOpen
  \bibfield  {author} {\bibinfo {author} {\bibfnamefont {S.}~\bibnamefont {Kay}}\ and\ \bibinfo {author} {\bibfnamefont {Y.~C.}\ \bibnamefont {Eldar}},\ }\href {\doibase 10.1109/MSP.2008.918027} {\bibfield  {journal} {\bibinfo  {journal} {IEEE Signal Processing Magazine}\ }\textbf {\bibinfo {volume} {25}},\ \bibinfo {pages} {133} (\bibinfo {year} {2008})}\BibitemShut {NoStop}%
\end{thebibliography}%



\begin{widetext}
    \medskip
    \begin{center}
        \large{\textbf{Supplemental Material}}
    \end{center}
    \medskip
    The supplemental material is organized as follows: First, we discuss the preliminaries related to biased estimation in \Cref{app:prelim}.
    Next, we provide a warm-up by analyzing the interaction-free case in \Cref{app:the-interaction-free}.
    Subsequently, we explain how we bound the errors in two distinct cases: Free Evolution (FE) case in \Cref{app:free-evolution} and Control-Enhanced (CE) case in \Cref{app:controlled-evolution}, respectively.
    At the end, we provide the details on the measurement in \Cref{app:detail-measurement} and the estimation procedure with CNOT interaction in \Cref{app:cnot-interaction}.
    
    \tableofcontents

    \section{Preliminaries}\label{app:prelim}
    
    \subsection{Biased Cram\'er-Rao lower bound}
    \label{app:biased-crb}
    The original Cram\'er-Rao lower bound (CRB) is valid only if the estimator is unbiased.
    In this work, the potential bias generated by the clock uncertainty makes the unbiased condition impossible to achieve, which necessitates considering biased CRB.
    Here we just provide a simple proof of the biased CRB (see Ref.~\cite{treeDetectionEstimationModulationTheory2001} for more details).
    \begin{corollary}
        If $\hat{\theta}$ is a biased estimator of $\theta$, i.e. $E(\hat{\theta})=\theta+b(\theta)$ with $b(\theta)$ being the bias, then
        \begin{align}
            \Var(\hat\theta)\ge \frac{\left(1+\partial b(\theta)/\partial \theta\right)^2}{\mathcal{F}_{\theta}}.
        \end{align}
    \end{corollary}
    \begin{proof}
        The proof is almost the same as the proof of the CRB. 
        Starting with the biased condition, we have
        \begin{align}
            \int d\mathbf{R} p(\mathbf{R}|\theta)\left[\hat\theta(\mathbf{R})-\theta-b(\theta)\right]=0.
        \end{align}
        By differentiating the above equation w.r.t. $\theta$, we have
        \begin{align}
            \int d\mathbf{R} \frac{\partial p(\mathbf{R}|\theta)}{\partial\theta}\left[\hat\theta(\mathbf{R})-\theta-b(\theta)\right]+\int d\mathbf{R} p(\mathbf{R}|\theta)(-1-\partial_\theta b(\theta))=0.
        \end{align}
        Applying the same trick as the proof of the CRB, we have
        \begin{align}
            \int d\mathbf{R} p(\mathbf{R}|\theta)\left[\hat\theta(\mathbf{R})-\theta-b(\theta)\right]\partial_\theta \ln p(\mathbf{R}|\theta)=1+\partial_\theta b(\theta).
        \end{align}
        Rewriting, we have
        \begin{align}
            \int d\mathbf{R} \left[\sqrt{p(\mathbf{R}|\theta)} \partial_\theta \ln p(\mathbf{R}|\theta)\right]\left[\sqrt{p(\mathbf{R}|\theta)}\left[\hat\theta(\mathbf{R})-\theta-b(\theta)\right]\right]=1+\partial_\theta b(\theta).
        \end{align}
        By applying the Cauchy-Schwarz inequality, we have
        \begin{align}
            \int d\mathbf{R} p(\mathbf{R}|\theta)\left[\hat\theta(\mathbf{R})-\theta-b(\theta)\right]^2\times\int d\mathbf{R} p(\mathbf{R}|\theta)\left[\partial_\theta \ln p(\mathbf{R}|\theta)\right]^2\ge \left(1+\partial_\theta b(\theta)\right)^2,
        \end{align}
        or equivalently
        \begin{align}
            \Var[\hat\theta]\ge \frac{\left(1+\partial b(\theta)/\partial \theta\right)^2}{\mathcal{F}_{\theta}},
        \end{align}
        which completes the proof.
    \end{proof}
    
    Another version of biased CRB is written in terms of mean squared error, which is given by
    \begin{align}
        E[(\hat\theta-\theta)^2]\ge \frac{\left(1+\partial b(\theta)/\partial \theta\right)^2}{\mathcal{F}_{\theta}}+b(\theta)^2,
    \end{align}
    where $E[(\hat\theta-\theta)^2]$ is the mean squared error.
    \begin{proof}
        According to the CRB, we have
        \begin{align}
            \Var[\hat\theta]=E[(\hat\theta-E(\hat\theta))^2]\ge \frac{\left(1+\partial b(\theta)/\partial \theta\right)^2}{\mathcal{F}_{\theta}},
        \end{align}
        which means that the biased estimator $\hat\theta$ is unbiased to the parameter $\theta+b(\theta)$.
        The LHS can be equivalently written as
        \begin{align}
            E[(\hat\theta-E(\hat\theta))^2]&=E[(\hat\theta-\theta-b(\theta))^2]\\
            &=E[(\hat\theta-\theta)^2-2b(\theta)(\hat\theta-\theta)+b(\theta)^2]\\
            &=E[(\hat\theta-\theta)^2]-2b(\theta)\times b(\theta)+b(\theta)^2\\
            &=E[(\hat\theta-\theta)^2]-b(\theta)^2,
        \end{align}
        where the first and third equalities are obtained by utilizing $b(\theta)=E(\hat\theta)-\theta$.
        Then we have
        \begin{align}
            E[(\hat\theta-\theta)^2]\ge \frac{\left(1+\partial b(\theta)/\partial \theta\right)^2}{\mathcal{F}_{\theta}}+b(\theta)^2.
        \end{align}
    \end{proof}
    
    \section{The interaction-free case}\label{app:the-interaction-free}
    In the absence of interaction, the Hamiltonian can be written as
    \begin{align}
        H = \omega \sigma_Z / 2,
    \end{align}
    and its corresponding unitary for evolution time $T$ is
    \begin{align}
        U(T) = \exp(-i \omega \sigma_Z T / 2).
    \end{align}
    Considering an input state $\rho$, due to the clock uncertainty, the final state is the weighted average in time:
    \begin{align}
        \rho' = \int du f(u) U[T(1+u)] \rho U[T(1+u)]^\dagger.
    \end{align}
    We can written the state $\rho$ in the eigenbasis of $\sigma_Z$:
    \begin{align}
        \rho' &= \int du f(u) U[T(1+u)]\sum_{jk} \rho_{jk} \ketbra{j}{k} U[T(1+u)]^\dagger\\
        &=\int du f(u) \sum_{jk} e^{-i\omega(\lambda_j-\lambda_k)(1+u)T/2}\rho_{jk}\ketbra{j}{k},
    \end{align}
    where $\{\lambda_j\}$ are the eigenvalues of $\sigma_Z$.
    It is straightforward that the optimal quantum Fisher information (QFI) can be achieved by choosing $\ketbra{+}$ as the input state, which yields the QFI
    \begin{align}
        e^{-2 \Re(\gamma)}T^2,
    \end{align}
    where $\gamma=-\ln \left[\int du f(u) e^{-i\omega u T}\right]$.
    In the parameter configuration of our numerical simulation, the QFI of the interaction-free case is close to $T^2$ within a small correction.
    
    In addition, the imaginary part of $\gamma$ can lead to additional bias, which cannot be eliminated by increasing the interrogation time or the number of repeat measurements.
    This bias can be understood by the correlation between frequency and time.
    To determine the frequency of some periodic event, intuitively, we need to estimate the phase accumulated in a certain time, then divide it by time to estimate the frequency.
    If the clock is subject to some uncertainty, it is unavoidable that the measurement of frequency becomes inaccurate.
    
    The systematic error is given by
    \begin{align}
        \lim_{\nu\rightarrow \infty}\mathcal{R}_{\text{interaction-free}}= O(\epsilon^2\omega^2),
    \end{align}
    where $\mathcal{R}_{\text{interaction-free}}$ is the mean squared error in the interaction-free case.
    Then we show how to obtain the result.
    We first assume we have already determined the value of $\phi=\omega T$.
    We estimate $\omega$ by estimating $\phi$ and then dividing the result by the time, i.e.,
    \begin{align}
        E(\hat{\omega})=\frac{E(\hat{\phi})}{\tilde{T}}=\frac{\omega T}{T+\Delta T},
    \end{align}
    where the $\Delta T$ is bounded by $\epsilon T$. Then we have the result
    \begin{align}
        \left[E(\hat\omega)-\omega\right]^2 \le \epsilon^2\omega^2,
    \end{align}
    or equivalently,
    \begin{align}
        \lim_{\nu\rightarrow \infty}\mathcal{R}_{\text{interaction-free}}= O(\epsilon^2\omega^2).
    \end{align}
    The result shows that this estimation protocol has a bias term that cannot be eliminated by either increasing the number of measurements or the duration of the process.
    We will show in the FE case, the error can be way larger than this bias term $\epsilon^2 \omega^2$, while the CE case can provide an error as small as this bias term.
    
    \section{The free evolution (FE) case}\label{app:free-evolution}
    
    \subsection{Optimal biased estimation}\label{app:opt-biased-est}
    
    The biased estimation in most cases is worse than the unbiased estimation. 
    However, if the derivative of the bias is negative, it is possible that the MSE in the biased estimation may be smaller than in the unbiased case.
    Although this reduction in MSE by biased estimation usually depends on the knowledge about unbiased estimator, we cannot simply ignore this possibility.
    
    We will use a technique to lower bound the MSE of biased estimation (see Ref.~\cite{kayRethinkingBiasedEstimationLectureNotes2008} for more details).
    Recall the biased CRB, we have
    \begin{align}
        \mathrm{E}\left[(\hat \omega_b-\omega)^2\right]\ge \frac{(1+\partial b/\partial \omega)^2}{\nu\mathcal{F}_\omega} + b^2,
    \end{align}
    where $\hat\omega_b$ is the biased estimator.
    The bias is $b:=\mathrm{E}(\hat \omega_b)-\omega$ by definition.
    Here we introduce the parameter $\omega_0$ as the reference value of $\omega$.
    The true value of $\omega$ is close to $\omega_0$.
    Naturally, we can expand the bias in power series at $\omega=\omega_0$,
    \begin{align}
        \mathrm{E}(\hat\omega_b) - \omega = c + m (\omega-\omega_0) + O[(\omega-\omega_0)^2],
    \end{align}
    where the zeroth order $c:=\mathrm{E}(\hat \omega_b)|_{\omega = 0}$.
    We let $c=0$, which means the bias estimator provides an accurate estimation when $\omega=\omega_0$, which only lifts the accuracy and relaxes the bound.
    It is important to note that the first derivative of the bias will influence the statistical error.
    For $|\omega-\omega_0|\ll 1$, we can ignore the higher-order terms and write the bias in the following form
    \begin{align}
        b=m(\omega-\omega_0).
    \end{align}
    By substitution, we have
    \begin{align}
        \mathrm{E}\left[(\hat\omega_b-\omega)^2\right]\ge \frac{(1+m)^2}{\nu\mathcal{F}_\omega} + m^2(\omega-\omega_0)^2.
    \end{align}
    We can find the minimum of the MSE by differentiating the MSE with respect to $m$ and set the result equal to zero, which yields
    \begin{align}
        m^* = -\frac{1}{1 + \nu\mathcal{F}_\omega (\omega_0 - \omega)^2},
    \end{align}
    and
    \begin{align}
        \min_{m}\mathrm{E}\left[(\hat\omega_b-\omega)^2\right] \ge \frac{1}{\nu \mathcal{F}_\omega+ (\omega -\omega_0)^{-2}}.
    \end{align}
    This result demonstrates that the estimation process can benefit from the bias.
    Notably, in the non-asymptotic case, when $\nu\mathcal{F}_\omega \ll (\omega-w)^{-2}$, the RHS of the equation simplifies to $(\omega-\omega_0)^{2}$.
    In the asymptotic case, where $\nu\mathcal{F}_\omega \gg (\omega-\omega_0)^{-2}$, this expression reduces to the unbiased Cramér-Rao bound (CRB).
    It is important to note that this minimization relies on both the unknown parameter $\omega$ and the unbiased estimator, allowing for the design of a biased estimator to enhance the estimation accuracy.
    In our study, we utilize this result to establish bounds on the Mean Squared Error (MSE).
    In the previous result, the only assumption is $|\omega-\omega_0|\ll 1$, which is reasonable given that the theorem of CRB is locally applicable.
    In subsequent analysis, we set $\omega_0=0$ and $|\omega|\ll 1$.
    Note that this assumption will not alter the overall results significantly.
    
    \subsection{The upper bound for the QFI in FE}\label{app:upper-bound-QFI-FE}
        The matrix representation of the Hamiltonian is
        \begin{align}
            H_{\rm tot}(\omega)=
            \left(
            \begin{array}{cccc}
             g+\frac{\omega }{2} & 0 & 0 & 0 \\
             0 & \frac{\omega }{2} & g & 0 \\
             0 & g & -\frac{\omega }{2} & 0 \\
             0 & 0 & 0 & g-\frac{\omega }{2} \\
            \end{array}
            \right).
        \end{align}
        By linear algebra, we can obtain the evolution operator $U=\exp(-i H_{\rm tot} T)$ by diagonalization.
        The evolution is given by
        \begin{align}
            U=
            \left(
            \begin{array}{cccc}
             e^{-\frac{1}{2} i T (2 g+\omega )} & 0 & 0 & 0 \\
             0 & \frac{e^{-i T \Omega } \left(\omega  (\omega +2 \Omega )+2 g^2 \left(1+e^{2 i T \Omega }\right)\right)}{2 \Omega  (\omega +2 \Omega )} & -\frac{i g \sin (T \Omega )}{\Omega } & 0 \\
             0 & -\frac{i g \sin (T \Omega )}{\Omega } & \cos (T \Omega )+\frac{i \omega  \sin (T \Omega )}{2 \Omega } & 0 \\
             0 & 0 & 0 & e^{-\frac{1}{2} i T (2 g-\omega )} \\
            \end{array}
            \right),
        \end{align}
        where $\Omega=\sqrt{g^2+\omega ^2/4}$.
        The environment is set as $\ket{0}$ initially.
        Therefore, we can obtain the Kraus operators:
        \begin{align}
            K_1&=\left(
                \begin{array}{cc}
                 e^{-\frac{1}{2} i T (2 g+\omega )} & 0 \\
                 0 & \cos (T \Omega )+\frac{i \omega  \sin (T \Omega )}{2 \Omega } \\
                \end{array}
                \right),\\
            K_2&=\left(
                \begin{array}{cc}
                 0 & -\frac{i g \sin (T \Omega )}{\Omega } \\
                 0 & 0 \\
                \end{array}
                \right),
        \end{align}
        and their derivatives w.r.t. $\omega$:
        \begin{align}
            \dot{K}_1&=\left(
                \begin{array}{cc}
                 -\frac{1}{2} i T e^{-\frac{1}{2} i T (2 g+\omega )} & 0 \\
                 0 & \frac{2 i T \omega ^2 \Omega  \cos (T \Omega )-\left(T \omega ^3+4 g^2 (T \omega -2 i)\right) \sin (T \Omega )}{16 \Omega ^3} \\
                \end{array}
                \right),\\
            \dot{K}_2&=\left(
                \begin{array}{cc}
                 0 & \frac{i g \omega  (\sin (T \Omega )-T \Omega  \cos (T \Omega ))}{4 \Omega ^3} \\
                 0 & 0 \\
                \end{array}
                \right).
        \end{align}
        Then we utilize the formula of extended channel QFI (see \cite{kolodynskiEfficientToolsQuantumMetrologyUncorrelatedNoise2013} for more details), which is given by
        \begin{align}
            \mathcal{F}_\omega^Q=4\min_h \left\Vert \sum_{j} \dot{\tilde{K}}_j^\dagger \dot{\tilde{K}}_j \right\Vert,
        \end{align}
        where $\dot{\tilde{K}}_j=\dot{K}_j-i\sum_{k} h_{jk}K_k$ and $\Vert\cdot\Vert$ denotes operator norm.
        To simplify the calculation, we can write the Kraus operators in the following form
        \begin{align}
            K_1&=\left(
                \begin{array}{cc}
                 b & 0 \\
                 0 & c \\
                \end{array}
                \right),\\
            K_2&=\left(
                \begin{array}{cc}
                 0 & a \\
                 0 & 0 \\
                \end{array}
                \right).
        \end{align}
        Let $A:=\sum_j \dot{\tilde{K}}_j^\dagger \dot{\tilde{K}}_j$, we can write down the following inequality to bound the QFI
        \begin{align}
            \mathcal{F}_\omega^Q= 4\min_h \sigma_{\rm max} \le 4\min_h \Tr(A),
        \end{align}
        where $\sigma_{\rm max}$ is the maximum eigenvalue of $A$. 
        The trace of $A$ is given by
        \begin{align}
            \Tr(A) = 2|h_{12}|^2 + |\dot{a} - i a h_{22}|^2 + |\dot{b} - i b h_{11}|^2 + |\dot{c} - i c h_{11}|^2.
        \end{align}
        The minimal trace can be obtained easily by derivative w.r.t. $h_{11},h_{12},h_{22}$ and setting the result to zero,
        \begin{align}
            \min_h \Tr(A) = \frac{4|a|^2 |\dot a|^2+(a^*\dot a-\dot a^* a)^2}{4|a|^2} + 
            \frac{4(|b|^2+|c|^2) (|\dot b|^2+|\dot c|^2)+(b^*\dot b-\dot b^* b+c^*\dot c-\dot c^* c)^2}{4(|b|^2+|c|^2)}.
        \end{align}
        By substituting the explicit form of $a,b,c,\dot{a},\dot{b},\dot{c}$, we can obtain the upper bound of the QFI
        \begin{align}
            \mathcal{F}_\omega^Q\le
            &\frac{2}{\left(4 g^2+\omega ^2\right)^2 \left(g^2 \cos (2 T \Omega )+3 g^2+\omega ^2\right)}\times
            \Bigl[
            8 g^6 T^2+g^4 \left(22 T^2 \omega ^2+8\right)\\
            &+g^2 \left(2 g^2+\omega ^2\right) \left(T^2 \left(4 g^2+\omega ^2\right)-4\right) \cos (2 T \Omega )+g^2 \omega ^2 \left(9 T^2 \omega ^2+4\right)\\
            &+8 g^4 T \sqrt{4 g^2+\omega ^2} \sin (2 T \Omega )+T^2 \omega ^6
            \Bigr].
        \end{align}
        We can obtain a more concise form by power series expansion:
        \begin{align}
            \mathcal{F}_\omega^Q&\le\frac{2 \left[g T \cos (T \Omega )+\sin (T \Omega )\right]^2}{g^2 (\cos (2 T \Omega )+3)} + O\left(\omega^2T^2/g^2\right)\\
            &=\frac{2\cos^2(\Omega T-\alpha)}{\left[\cos(2\Omega T)+3\right]}\left(T^2+\frac{1}{g^2}\right)+\frac{\omega^2}{g^2} T^2 +O\left(\frac{\omega^4 T^2}{g^4}\right),
        \end{align}
        where $\alpha=\tan^{-1}[1/(gT)]$.
        For sufficient large $T$, we can write
        \begin{align}
            \mathcal{F}_\omega^Q \le \frac{2T^2\cos^2(\Omega T)}{\cos(2\Omega T)+3} + \frac{\omega^2}{g^2}T^2 +O(\omega^4 T^2/g^4),
        \end{align}
        which is the formula that we have shown in the main text.
        If the local minima of the bound are achievable, i.e., the leading term equal to zero, then we have
        \begin{align}
            \min_{f\in\mathcal{N_\delta}}\mathcal{F}_\omega^Q\left[f,\fe\right] \le \frac{\omega^2}{g^2} T^2 + O(\omega^4 T^2/g^4),
        \end{align}
        which needs $\epsilon\ge \epsilon^\ast$ to guarantee the achievability of the local minima.
        If the local minima is not achievable, i.e., $\epsilon < \epsilon^\ast$, we have
        \begin{align}
            \min_{f\in\mathcal{N}_\delta}\mathcal{F}_\omega^Q\le T^2 \left[1+\frac{\omega^2}{g^2}-\frac{2}{\cos (2 T \Omega  \epsilon )+3}\right]+O(\omega^4 T^2/g^4).
        \end{align}
    
        Correspondingly, by utilizing the bound for biased estimation shown in \Cref{app:opt-biased-est}, if $\epsilon\ge\epsilon^\ast$, we have
        \begin{align}
            \lfe\ge \frac{g^2 \omega ^2}{g^2+\nu T^2 \omega ^4}+O\left[{\nu\omega^8 T^2}/{(g^2+\nu T^2\omega^4)^2}\right],
        \end{align}
        and if $\epsilon< \epsilon^\ast$, we have
        \begin{align}
            \lfe\ge\frac{1}{\nu  T^2 \left[\frac{\omega ^2}{g^2}+\cos ^2(\epsilon T \Omega )\right]+\frac{1}{\omega ^2}} + O\bm{(}{\omega ^4}/\{{\nu  T^2 [g^2 \cos ^2(T \Omega  \epsilon )+\omega ^2]^2}\}\bm{)}.
        \end{align}
    
    \section{The control-enhanced (CE) case}\label{app:controlled-evolution}
    
    In this section, we use the same assumptions as previously mentioned, which is $|\omega|\ll g$ and the environment is initialized in $|0\rangle$.
    
    \subsection{Evolution with perfect control}\label{app:evolution-with-perfect-control}
    \begin{figure}
        \centering
        \includegraphics[width=0.8\textwidth]{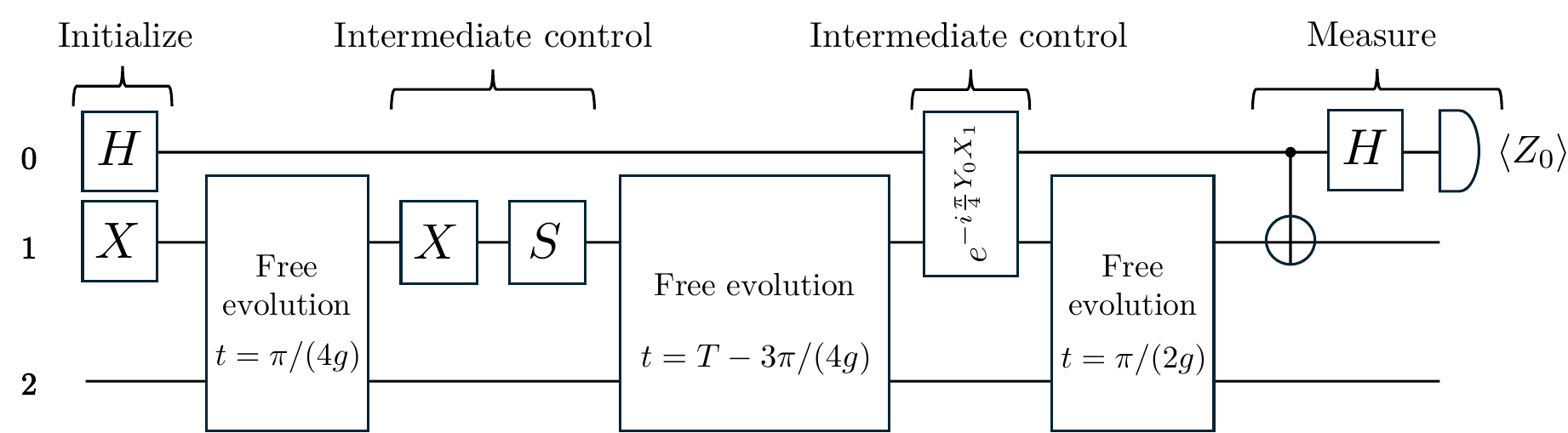}
        \caption{The detailed control strategy for frequency estimation. The qubit of ancilla, system, and environment are labeled by 0, 1, and 2 respectively. All state is set to be $\ket{0}$ at the beginning. The gates $H,~S,$ and$~X$ in the circuit are basic quantum logic gates, representing Hadamard, $\pi/2$-phase, and Pauli-$X$ gate respectively.}
        \label{fig:detailed-ctrl}
    \end{figure}
    For clarity, we first assume that everything is ideal and discuss the imperfections later.
    To proceed, we label the qubit of ancilla, system, and environment by 0, 1, and 2 respectively.
    As shown in \Cref{fig:detailed-ctrl}, we begin with the state $\ket{+ 1}$, which is actually on the joint space of the ancilla and system.
    We will first focus on the joint space of the system and environment, where the state is $\ket{10}$ currently.
    We then convert the state into $e^{-i\text{SWAP}\pi/4}|10\rangle=|01\rangle+i|10\rangle$ by utilizing the interaction between system and environment, which costs time $\pi/(4 g)$.
    By applying an additional control $S_1X_1$ as shown in \Cref{app:detail-measurement}, the state $|01\rangle+i|10\rangle$ can be transformed into $|00\rangle+|11\rangle$, which is the starting point for the following discussion.
    The ancilla is prepared into $\ket{+}$ by the Hadamard gate. Then the current quantum state can be written as
    \begin{align}
        |+\rangle(|00\rangle+|11\rangle).
    \end{align}
    We let this state evolve freely according to the Hamiltonian $H=g\text{SWAP}_{12}+\omega Z_1/2$.
    We noticed that the above state is an eigenstate of SWAP.
    As all the states generated by the rotation in system Hamiltonian still lie on the eigenspace of the interaction Hamiltonian, where the interaction Hamiltonian has zero effect on the state.
    Thus, we end up with such a state after evolution time $T'=T-\frac{3\pi}{4g}$,
    \begin{align}
        |+\rangle(|00\rangle+e^{-i\omega T'}|11\rangle).
    \end{align}
    Here we apply an additional control $\exp(-iY_0 X_1\pi/4)$ on the joint space of the ancilla and the system. By linear algebra, we obtain the state
    \begin{align}
        (|00\rangle+|10\rangle-|01\rangle+|11\rangle)|0\rangle+e^{-i\omega T'}(|01\rangle+|11\rangle-|00\rangle+|10\rangle)|1\rangle.
    \end{align}
    At this point, if we trace out the third qubit (environment), the phase factor containing the parameter $\omega$ will get canceled. 
    Thus, this gives us no information about $\omega$. 
    To tackle this, we allow the state to evolve freely subject to the Hamiltonian for an appropriate time, i.e., $\pi/(2 g)$ such that the evolution approximates a full SWAP.
    Then we have
    \begin{align}
        (|00\rangle+|10\rangle-e^{-i\omega T'}|01\rangle+e^{-i\omega T'}|11\rangle)|0\rangle+(e^{-i\omega T'}|01\rangle+e^{-i\omega T'}|11\rangle-|00\rangle+|10\rangle)|1\rangle.
    \end{align}
    By tracing out the environment and rewriting it into density matrix form, we have
    \begin{align}
        \rho=\left(
            \begin{array}{cccc}
                \frac{1}{4} & -\frac{1}{4} e^{i T' \omega } & 0 & 0 \\
                -\frac{1}{4} e^{-i T' \omega } & \frac{1}{4} & 0 & 0 \\
                0 & 0 & \frac{1}{4} & \frac{1}{4} e^{i T' \omega } \\
                0 & 0 & \frac{1}{4} e^{-i T' \omega } & \frac{1}{4} \\
            \end{array}
            \right).
    \end{align}
    We can measure the output state in the following basis
    \begin{align}
        \{\ket{0}\ket{+},\ket{0}\ket{-},\ket{1}\ket{+},\ket{1}\ket{-}\},
    \end{align}
    which generates the following probabilities
    \begin{align}
        &p_1=p_4=\Tr\left(\ket{0,+}\bra{0,+}\rho\right)=\Tr\left(\ket{1,-}\bra{1,-}\rho\right)=\frac{1}{2}\sin^2\left(\frac{\omega T'}{2}\right),\\
        &p_2=p_3=\Tr\left(\ket{1,+}\bra{1,+}\rho\right)=\Tr\left(\ket{0,-}\bra{0,-}\rho\right)=\frac{1}{2}\cos^2\left(\frac{\omega T'}{2}\right),
    \end{align}
    which is actually a Boolean-valued outcome. 
    The Fisher information associated with the probabilities is given by
    \begin{align}
        \mathcal{I}_\omega&=\sum_j \frac{1}{p_j}\left(\frac{\partial p_j}{\partial \omega}\right)^2\\
        &=T'^2 = \left(T-\frac{3\pi}{4g}\right)^2.
    \end{align}
    which is close to the optimal QFI.

    \subsection{Bound the error in CE case}
    \subsubsection{Bound the systematic error}\label{app:bound-bias-CE}
    In our setup, the imperfections can be modeled as sum of unitaries. 
    We first give a general discription and then analyze a specific case for demonstration.
    
    A sequence of unitaries $U_i$ is affected by some imperfection $\epsilon_i V_i$ ($\epsilon_i \ge 0~\forall~i$), which can be written as
    \begin{align}
        U=\prod_{i} U_i(1+\epsilon_i V_i)+\mathbf{O}(\eta^2)=\mathcal{U}_0+\sum_{i}\epsilon_i\mathcal{V}_i+\mathbf{O}(\eta^2),
    \end{align}
    where $\mathcal{U}_0$ and $\mathcal{V}_i$ represent the zeroth and first order of unitaries in the series expansion.
    $\mathcal{V}_i$ are also unitary.
    The bold big-O notation is for matrix form and $\eta$ is defined by $\eta:=\sum_i\epsilon_i$.
    Specifically, the zeroth order $\mathcal{U}_0=\prod_i U_i$ refers to the operation in the noiseless scenario.
    If we prepare a product initial state $\ket{\psi_0}$ then applying the unitary, and trace out the environment at the end, we have
    \begin{align}
        \rho&=\Tr_{\rm E}\left[U\ket{\psi_0}\bra{\psi_0}U^\dagger\right]\\
        &=\rho_0+\Tr_{\rm E} \left[\sum_i\epsilon_i\mathcal{V}_i\ket{\psi_0}\bra{\psi_0}\mathcal{U}_0 + h.c. \right] + \mathbf{O}(\eta^2),
    \end{align}
    where $\rho_0$ is the zeroth order density matrix that can be written as $\Tr_{\rm E}\left[\mathcal{U}_0\ketbra{\psi_0}\mathcal{U}_0^\dagger\right]$.
    Then we can bound first order error by
    \begin{align}
        \Vert \rho - \rho_0\Vert \le 2\eta + {O}(\eta^2).
    \end{align}
    The norm denotes the operator norm. 
    This is a straightforward application of the triangular inequality.
    The reason why we exploit the operator norm is that it provides an upper bound for the error in the output probabilities.
    For projective measurement $M=\ketbra{\Psi}$, the output probability is given by $\Tr(\rho \ketbra{\Psi})$.
    The upper bound of the distance in the output probability can be obtained by eigendecomposition in eigenbasis of the operator $(\rho-\rho_0)$ that we write as $A$ for short:
    \begin{align}\label{eq:bound-prob-err-by-op-norm-rho}
        \left| \braket{\Psi|A|\Psi} \right| = \left| \sum_{i,j} \Psi_i^*\Psi_j \bra{a_i} \left(\sum_{k}\lambda_k\ketbra{a_k}\right)\ket{a_j} \right| = \left|\sum_k \Psi_k^*\Psi_k  \lambda_k \right| \le \sum_k \Psi_k^*\Psi_k\left|\lambda_k \right| \le \max_k |\lambda_k| \le \Vert A\Vert.
    \end{align}
    Then we have $|p-p_0|\le\Vert \rho-\rho_0\Vert$ where $p$ and $p_0$ are the output probabilities of $\rho$ and $\rho_0$ respectively.
    
    For illustration, we show how this model is useful for our setup in a specific example.
    In the CE case, we have several consequences of the control imperfection.
    One of them is the evolution of the system during applying the approximate SWAP.
    According to \Cref{app:evolution-with-perfect-control}, the whole controlled evolution can be written as following form
    \begin{align}
        U=U_1 e^{i\delta_1}U_2 e^{i\delta_2} &= U_1
        \left(\mathds{1}\cos(\omega\pi/8g)-iZ_1\sin(\omega\pi/8g)\right)U_2\\
        &\times\left(\mathds{1}\cos(\omega\pi/4g)-iZ_1\sin(\omega\pi/4g)\right)+\mathbf{O}(\omega^2/g^2)\\
        &=U_1U_2+\frac{\omega\pi}{8g}(V_1+2V_2)+\mathbf{O}(\omega^2/g^2),
    \end{align}
    where $V_{1,2}$ are two unitary operators.
    Then we can upper bound the leading-order error
    \begin{align}
        \frac{|\omega|\pi}{8g}(\Vert U_1+2U_2\Vert) \le\frac{\pi|\omega|}{8g}\left(\Vert U_1\Vert+2\Vert U_2\Vert\right) =\frac{3\pi|\omega|}{8g},
    \end{align}
    where the second inequality comes from the triangle inequality and $\Vert\cdot\Vert$ denotes the operator norm.
    By this analyzation, we see that the total error is the sum of each individual error term, which verifies the mentioned technique.
    We will consistently use this trick in the following calculation.
    From the result, we find that this imperfection will contribute an error $\epsilon_{1}=3\pi|\omega|/(8g)$.
    
    Another important source of error is the free precession in the middle part.
    In this stage, the state is almost in the degenerate eigen-subspace (DES) of SWAP.
    Therefore the dynamics of SWAP will only affect the part of the state outside DES, and the state in DES will be only affected by the dynamics of system evolution due to clock uncertainty.
    Specifically,
    \begin{align}
        \ket{\psi_{\rm after}}&=\e^{-\i H T'}\ket{\psi_{\rm before}}\\
        &=\e^{-\i H T'}\left(\ket{\psi_{\rm DES}}+\epsilon'\ket{\psi_{\rm DES}^{\perp}}\right)\\
        &\approx \e^{-\i \omega Z_1 T' / 2}\ket{\psi_{\rm DES}}+\epsilon'\e^{-\i g \textrm{SWAP}T'}\ket{\psi_{\rm DES}^{\perp}}.
    \end{align}
    We know that the error part is due to the imperfection of state preparation and controls before the free precession.
    This imperfection can be quantified by the error term
    \begin{align}\label{eq:app:eps-prime}
        \epsilon'=\frac{\epsilon\pi}{4} + \epsilon(\phi_{\rm c0}+\phi_{\rm c1}),
    \end{align}
    where the two terms represent the errors from clock uncertainty in approximating the partial SWAP, and the clock uncertainty in state preparation and first control.
    We will show how to obtain this result.
    Before we proceed, we first consider a quantum gate with a phase factor $\varphi$, i.e.,
    \begin{align}
        \text{GATE}=e^{-i A\varphi},
    \end{align}
    where $A$ is Hermiian and $\Vert A\Vert=1$.
    If the gate is applied with clock uncertainty, e.g., $\epsilon\ll 1$, we have
    \begin{align}
        \text{GATE}^{1+\epsilon}=e^{-i A\varphi (1+\epsilon)}=e^{-i A\varphi}(\mathds{1}-i A\epsilon \varphi),
    \end{align}
    and the leading-order error term can be bounded by
    \begin{align}
        \Vert\text{GATE}^{1+\epsilon}-\text{GATE}\Vert \le \Vert A \epsilon\varphi\Vert= |\epsilon\varphi|.
    \end{align}
    From this result, we can conclude that the error term contributed by a rotation gate is propotional to the phase associated with the gate.
    As shown in \Cref{eq:app:eps-prime}, the terms $\phi_{\rm c0}$ and $\phi_{\rm c1}$ are phases in the rotation gates of the state preparation and the first control.
    In our control scheme, we assume the state initially is $\ket{00}$.
    As shown in \Cref{fig:detailed-ctrl}, we need to prepare the state into $\ket{+1}$.
    To finish this, we need the phase at most $\pi/2$.
    The second control is $S X$, and we need the phase at most $\pi/2+\pi/4$, where $\pi/4$ for $S=\exp(-iZ\pi/4)$ and $X=\exp(-iX\pi/2)$.
    In total, $\phi_{\rm c0}+\phi_{\rm c1}=\frac{5\pi}{4}$.
    Combined with the error in free precession, these error terms can be summarized to $\epsilon_2=\epsilon|\omega| T'/2 + \epsilon'$.
    
    We can do a similar analysis for the remaining controls and evolutions.
    First, as illustrated in \Cref{fig:detailed-ctrl}, the second controls need phase $\pi/4$ and the final readout control needs to rotate the state from $Z\otimes X$ basis to computational basis $Z\otimes Z$ so that it needs phase $\pi/2$.
    Therefore, combined with the error in the evolution, we have
    \begin{align}
        \epsilon_3=\frac{\epsilon \pi}{2} + \frac{3\pi\epsilon}{4}.
    \end{align}
    Finally, we can combine all of the error terms to get
    \begin{align}
        \eta = \epsilon_1 + \epsilon_2 + \epsilon_3 = \frac{3\pi|\omega|}{8g}+\frac{11\pi\epsilon}{4} + \frac{\epsilon|\omega| T'}{2}.
    \end{align}
    In density matrix form, the error term is bounded by
    \begin{align}
        \Vert \tilde\rho -\rho\Vert\le 2\eta + O(\eta^2).
    \end{align}
    where $\tilde\rho$ is the state with the error term and the first inequality comes from the triangle inequality.
    In terms of probability, we can use the property shown in \Cref{eq:bound-prob-err-by-op-norm-rho} to get
    \begin{align}
        |\tilde{p}-p|=\left|\Tr\left[(\tilde\rho-\rho) M^\dagger M\right]\right|\le \Vert\tilde\rho-\rho\Vert\le 2\eta + O(\eta^2).
    \end{align}
    where $M$ is the measurement operator.
    
    The above analysis shows how the error in control affects the probabilities of measurement outcomes.
    Note that this analysis does not rely on specific form of distribution, meaning for any bounded distribution of the time uncertainty, the conclusion is still valid.
    This differs from the FE case where the measurement outcome is sensitive to the distribution.
    
    Furthermore, if $p=\cos^2(\omega T')$ and we use the function $\hat\omega(p)=\frac{1}{T'}\cos^{-1}(2p-1)$ to estimate the frequency, the bias is then given by
    \begin{align}
        |b(\omega)|&=\left|\frac{\partial f}{\partial p}\right|\times|\tilde{p}-p|+O\left(|\tilde{p}-p|^2\right)\\
        &\le \frac{c}{T'}\times 2\eta + O(\eta^2)\\
        &= \frac{2c}{T'}(\frac{3\pi\omega}{8g}+\frac{11\pi\epsilon}{4} + \frac{\epsilon\omega T'}{2}).
    \end{align}
    where $c=\left|\partial_p \cos^{-1}(2p-1)\right|$ and $T'=T-\frac{3\pi}{4g}$.
    As we know that $2p - 1=\cos(\omega T')$, we have
    \begin{align}
        |b(\omega)| = \frac{4 |\csc(\omega T')|\eta}{T'}
    \end{align}
    
    \subsubsection{Upper bound of MSE}\label{app:upper-bound-Lce}
    We propose a detailed control scheme, with certain input state, controls, and measurement.
    Therefore, we can directly calculate the MSE of the CE case.
    We denote the estimator by $\hat \omega(\mathbf{x})$, where $\mathbf{x}$ is the measurement data.
    The MSE can be written as
    \begin{align}
        \mathrm{E}\left[(\hat\omega-\omega)^2\right]=\Var(\hat\omega)+\left[\mathrm{E}(\hat\omega)-\omega\right]^2.
    \end{align}
    The second term is the bias that has been discussed in the last section.
    We merely focus on the variance here.
    Let $\mathbf{x}=(x_1, x_2, ..., x_\nu)$ be the measurement outcomes.
    Concretely, we measure the computational basis and assign the $x_i$ with different values as follows
    \begin{align}
        &x_i = -1, &\text{if the outcome is spin up.}\\
        &x_i = 1,  &\text{if the outcome is spin down.}
    \end{align}
    The statistical fluctuation of the measurement outcomes is the only source that contributes to the variance.
    There are many ways of combining the measurement data to give a prediction for the desired parameter $\omega$, but here we consider the average of $\{x_i\}$, i.e.,
    \begin{align}
        \hat\omega(\mathbf{x})=\hat\omega(\bar{x})=\frac{1}{T'}\cos^{-1}\left(\bar x\right),
    \end{align}
    where $\bar x:=\sum_{i=1}^\nu x_i/\nu$.
    According to error propagation, the variance of $\hat\omega$ can be written as
    \begin{align}
        \Var(\hat\omega) &= \left|\frac{\partial\hat\omega}{\partial\bar x}\right|^2\Var(\bar x)\\
        &=\frac{1}{T^2(1-\bar x^2)}\Var(\bar x)\label{app:eq:error-propagation}.
    \end{align}
    The another term $\Var(\bar x)$ can be written as
    \begin{align}
        \Var(\bar x) = \frac{1}{\nu^2} \Var\left(\sum_i x_i\right) = \frac{1}{\nu}\Var(x).
    \end{align}
    The last equality comes from the i.i.d. condition.
    In our model, we can consider the following probability distribution, for each single-shot measurement, the probability of spin pointing up or down can be represented as follows
    \begin{align}
        &\Pr(x=1) = p + \bar\alpha,\\
        &\Pr(x=-1) = 1 - p - \bar\alpha.
    \end{align}
    Here $\bar\alpha$ is the averaged shift of the original probability caused by the clock uncertainty and can be written as
    \begin{align}
        \bar\alpha = \int du \alpha(u) f(u),
    \end{align}
    where $\alpha(u)$ is the shift of the probability in one-shot measurement.
    As we have shown in the last section, $|\alpha|$ is bounded within $[0,2\eta]$.
    Then we have the overall variance of $x$:
    \begin{align}
        \Var(x) &= \mathrm{E}(x^2) - [\mathrm{E}(x)]^2\\
        &= 1 - \left[(2p-1)+2\bar\alpha\right]^2\\
        &= 1- (2p-1)^2 + O(\eta).
    \end{align}
    Substituting all to the \Cref{app:eq:error-propagation}, we have
    \begin{align}
        \Var(\hat \omega) &= \frac{1}{\nu T^2} \times\frac{1-[\mathrm{E}(x)]^2}{1-\bar x^2}\\
        &\approx\frac{1}{\nu T^2} \times \frac{1-(2p-1)^2 + O(\eta)}{1-(2p-1)^2+O(\eta)}\\
        &=\frac{1+O(\eta)}{\nu T^2},
    \end{align}
    where in the second equality we make the approximation $\bar x \approx \mathrm{E}(x)$. This approximation is correct for sufficiently large $\nu$.
    Combining the result in the last section, we have the upper bound for the $\lce$:
    \begin{align}
        \lce \le \frac{1+O(\eta)}{\nu T'^2} + \frac{16\csc^2(\omega T')\eta^2+O(\eta^3)}{T'^2},
    \end{align}
    where $\eta=\frac{3\pi|\omega|}{8g}+\frac{11\pi\epsilon}{4}+\frac{\epsilon|\omega| T'}{2}$.
    This completes the proof of Theorem 1 in main text.
    
    \section{Details on measurements}\label{app:detail-measurement}
    
    The Hamiltonian can be expressed as $H=g\mathrm{SWAP}_{12}+\omega Z_1/ 2$, where indices $0$, $1$, and $2$ denote the ancilla, system, and environment, respectively. This notation will be consistently used throughout this and next section.
    Our analysis will still focus on the strong-coupling regime, i.e., $g\gg \omega$.
    
    \medskip
    
    \noindent\textbf{Free evolution estimation procedure.}
    \begin{itemize}
        \item Initialize the system-ancilla state to the Bell state $\ket{00}+\ket{11}$.
        \item Measure at time $T$, which satisfies $\exp(-igT\mathrm{SWAP})=\mathds{1}\Leftrightarrow\sin(gT)^2=0$.
        \item Estimate the expectation value $\braket{O_{\rm FE}}$ of the observable $O_{\rm FE}=\mathrm{CNOT}_{01} X_0 \mathrm{CNOT}_{01}$.
        \item Estimate the frequency as $\hat\omega/2 = \cos^{-1}(\braket{O_{\rm FE}})/T$, where the factor $1/2$ is the correction of Zeno dynamics.
    \end{itemize}
    
    We will now demonstrate the optimality of this protocol by calculating the Quantum Fisher Information (QFI) of the measurement procedure and deriving an upper bound that coincides with it. This proof assumes the absence of clock uncertainty, as incorporating unknown clock uncertainty would make finding an optimal strategy for generic distributions nearly impossible.
    
    \begin{proof}
    In the case without clock uncertainty, i.e., $\epsilon=0$, 
    the final state in the protocol can be readily determined as:
    \begin{align}
        \ket{\psi}=(\ket{00}+e^{-i\omega T/2}\ket{11})/\sqrt{2}
    \end{align}
    The QFI of the above protocol is given by:
    \begin{align}
        \mathcal{F}^Q_{\omega} &= 4\left(\braket{\partial_\omega \psi|\partial_\omega \psi}-|\braket{\partial_{\omega}\psi|\psi}|^2\right)\\
        &= T^2/4.
    \end{align}
    
    We implement the same techique that we have used in \Cref{app:upper-bound-QFI-FE} to derive another upper bound of the QFI in the FE case by selecting a specific $h$:
    \begin{align}
        \mathcal{F}^Q_{\omega}=4\min_h \left\Vert \sum_{j} \dot{\tilde{K}}_j^\dagger \dot{\tilde{K}}_j \right\Vert \le 4 \left\Vert \sum_{j} \dot{\tilde{K}}_j^\dagger \dot{\tilde{K}}_j \right\Vert.
    \end{align}
    Equivalently:
    \begin{align}
        \mathcal{F}^Q_{\omega}\le 4 \sigma_{\rm max},
    \end{align}
    where $\sigma_{\rm max}$ is the largest singular value of the matrix $A=\sum_{j} \dot{\tilde{K}}_j^\dagger \dot{\tilde{K}}_j$.
    By setting $h_{12}=h_{21}=h_{22}=0$ and $h_{11}=-T/4$, we obtain a diagonalized $A$ as follows:
    \begin{align}
        A=T^2\begin{pmatrix}
            1/16 & 0\\
            0 & \frac{1}{32} \left[\cos (T \Omega' )+1\right]+O(\omega^2/g^2)
        \end{pmatrix},
    \end{align}
    where $\Omega'=\sqrt{4g^2+\omega^2}$.
    The largest singular value of $A$ is $T^2/16$.
    Consequently, we have:
    \begin{align}
        F\le 4\sigma_{\rm max}=T^2/4,
    \end{align}
    which matches the QFI of the protocol, thus proving its optimality.
    \end{proof}
    
    \medskip
    \noindent\textbf{Controlled evolution estimation procedure}
    \medskip
    
    The following procedure after the process shown in \Cref{fig:detailed-ctrl} is listed below:
    \begin{itemize}
        \item Measure the expectation value $\braket{O_{\rm CE}}$ of the observable $O_{\rm CE}=Z_0 X_1$.
        \item Estimate the frequency as $\hat \omega = \cos^{-1}(\braket{O_{\rm CE}})/T$.
    \end{itemize}
    
    \begin{figure}
        \centering
        \includegraphics[width=0.8\linewidth]{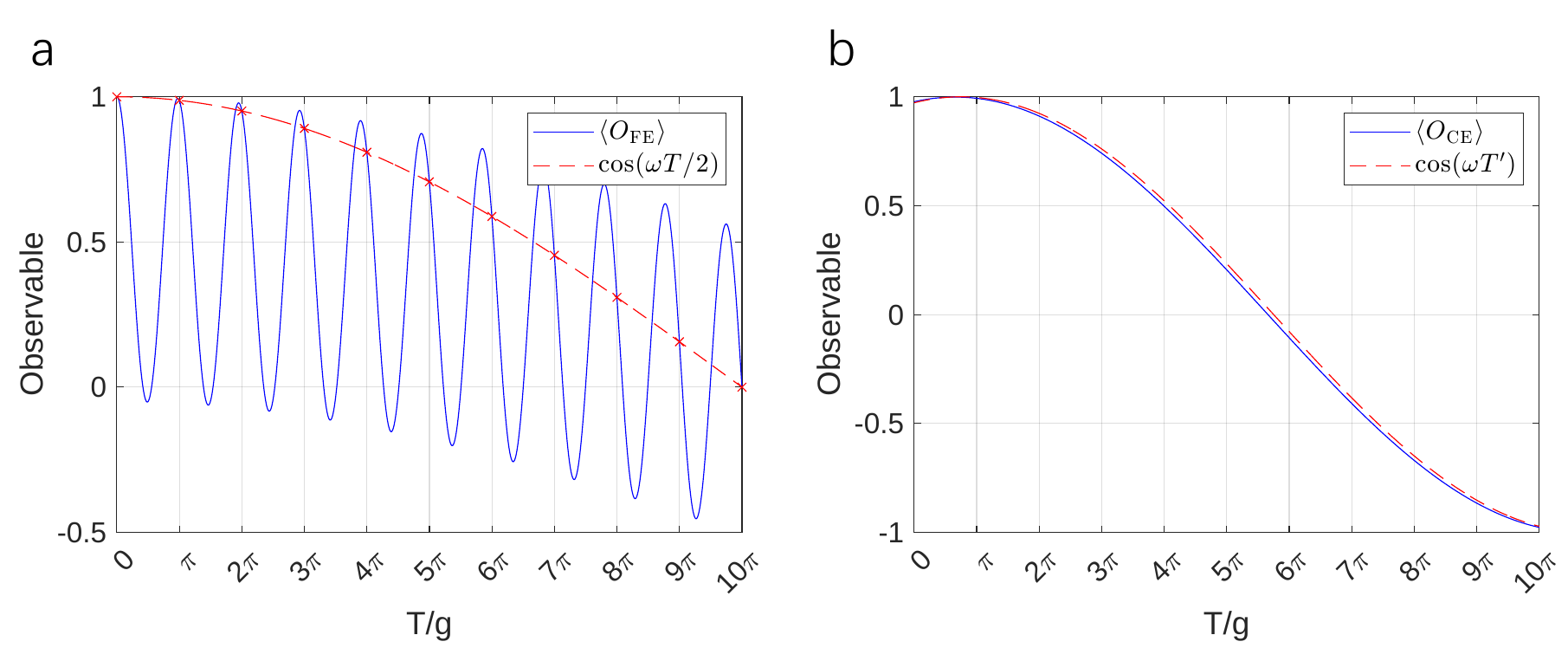}
        \caption{\textbf{The observables evolve over time.} The observables $O_{\rm FE}$ and $O_{\rm CE}$ serve as estimators of the unknown frequency $\omega$, providing optimal estimation efficiency in FE and CE case, respectively. \textbf{(a)} illustrates the observable in the FE case, which exhibits rapid fluctuations due to the dynamical evolution of the interaction Hamiltonian. The red crosses indicate the points where the Quantum Fisher Information (QFI) reaches its maximum, corresponding to the value of $\cos(\omega T /2)$. \textbf{(b)} depicts the controlled evolution.}
        \label{fig:obs-T}
    \end{figure}
    
    To clearly demonstrate the distinction between these two protocols, we plot the time-dependent observables in \Cref{fig:obs-T}.
    In \Cref{fig:obs-T}~(a), the system must be measured at a series of discrete time points (small red crosses) to obtain a precise estimation of $\omega$. 
    Due to the rapid fluctuations, the requirements for clock precision become significantly more stringent.
    In contrast, as shown in \Cref{fig:obs-T}~(b), the observable changes slowly and reflects a rate of change that depends solely on the desired parameter $\omega$, providing greater robustness against clock uncertainty.
    Therefore, we can conclude that our control strategy effectively mitigates the estimation error in the presence of clock uncertainty.
    
    \section{Supplemental details on CNOT interaction}\label{app:cnot-interaction}
    
    The Hamiltonian is given by $H=\omega Z_1/2+g \text{CNOT}_{12}$ and the free evolution can be written as $e^{-iHt}$. 
    In this section, we also focus on the strong-coupling regime, i.e., $g\gg \omega$.
    In the FE case, since the CNOT operator commutes with the Pauli Z, there is no Zeno dynamics. Thus, we can apply measurement at the times of decoupling to recover the precision, which is the same as demonstrated in the previous section with only a change in the correction of Zeno dynamics. We obtain
    \begin{align}
        \hat\omega = \cos^{-1}(\braket{O_{\rm FE}})/T.
    \end{align}
    The optimality is proven by calculating the Fisher information, which is $T^2$, coinciding with the interaction-free case.
    
    The CE strategy provided is quite complicated. To illustrate, we introduce a parameterized circuit $U$ defined in \Cref{fig:vqa-ansatz}. This is a parameterized quantum circuit with six parameters, where the rotation gates are defined as
    \begin{align*}
        &R_X(\theta):=e^{-i\theta X/2},~R_Y(\theta):=e^{-i\theta Y/2},\\
        &R_{XX}(\theta):=e^{-i \theta X\otimes X},~R_{ZZ}(\theta):=e^{-i \theta Z\otimes Z}.
    \end{align*}
    
    We utilize this parameterized circuit for the state preparation and the read-out transformation as shown in \Cref{fig:cnot-config}~(a). The detailed parameters are listed in \Cref{fig:cnot-config}~(b). At the end of the circuit, we apply measurement to calculate the expectation value of the observable $O(\rm CE)$. Then, we apply a function that maps the expectation value of the observable to the desired parameter $\omega$. The explicit form of the function is given below:
    \begin{align}
        \hat\omega_{\rm CE}(\bar x)= 0.99443 \left[2\pi-\cos^{-1}(x/0.499971519) + 2\pi K-197.427242\right]/T,
    \end{align}
    where $K:=\lfloor (\omega T+197.427242)/(2\pi)\rfloor$ is assumed to be known, and $\bar x$ is the averaged measurement outcome. The function is obtained by curve fitting as its analytical form can be much more complicated.
    
    \begin{figure}
        \centering
        \includegraphics[width=0.6\linewidth]{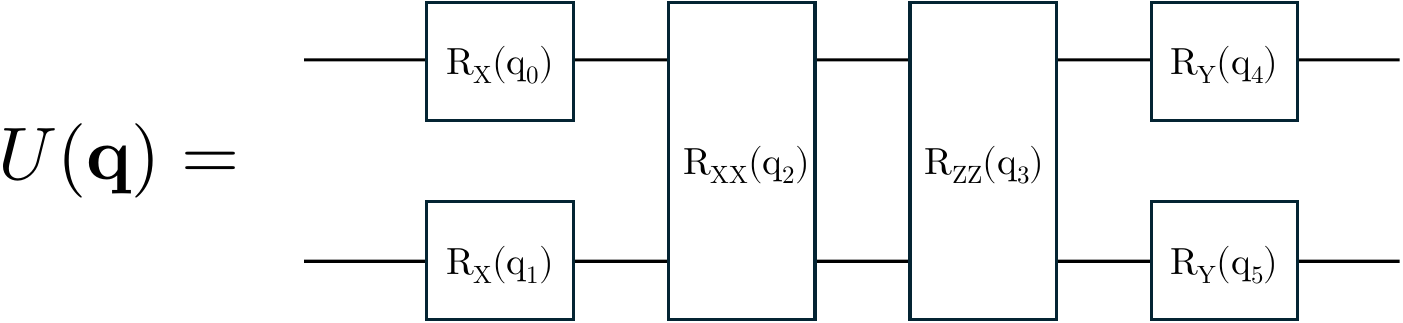}
        \caption{\textbf{Circuit representation of the ansatz.}}
        \label{fig:vqa-ansatz}
    \end{figure}
    
    \begin{figure}
        \centering
        \includegraphics[width=0.8\linewidth]{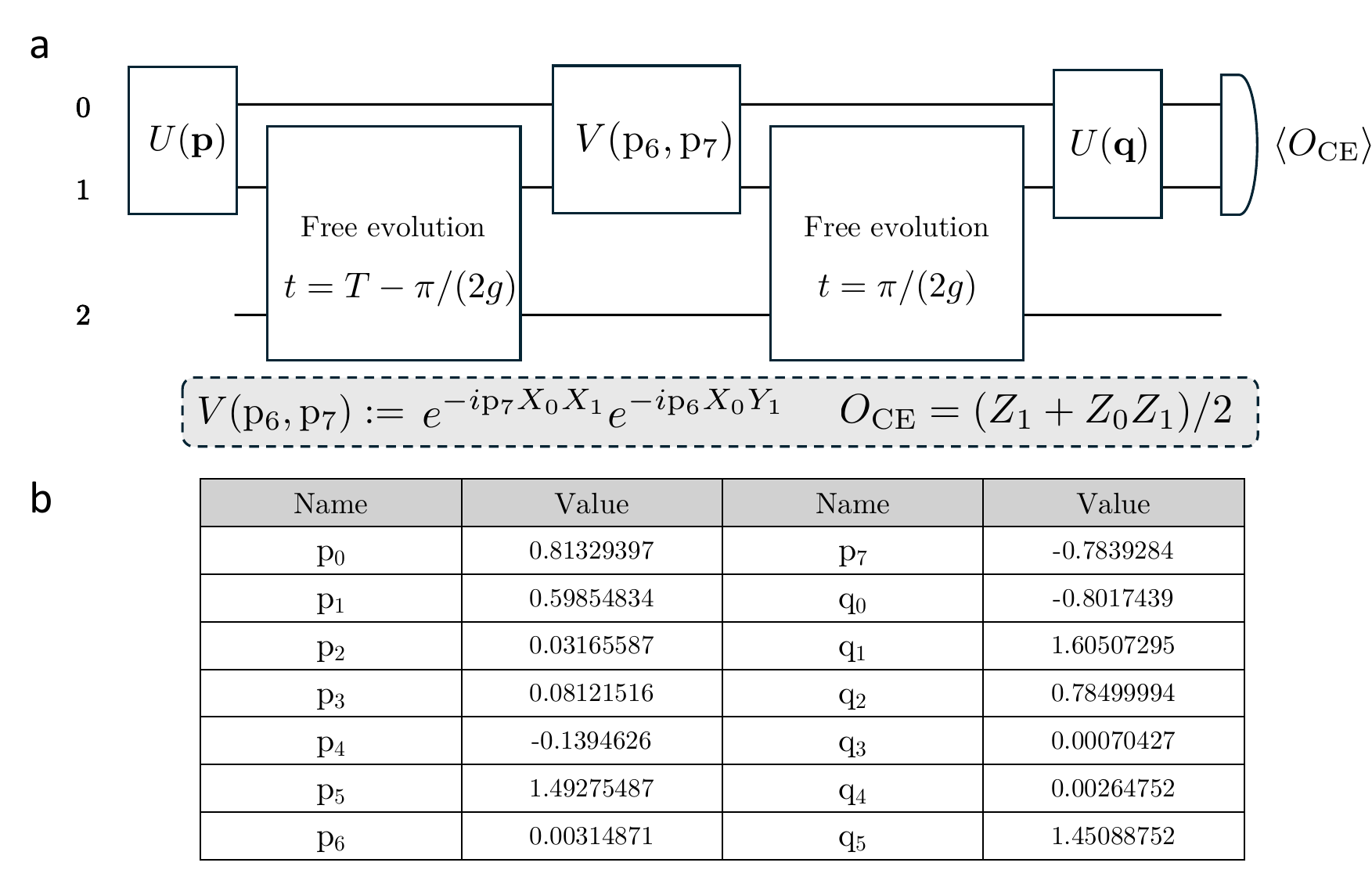}
        \caption{\textbf{Detailed configuration of the estimation protocol.} \textbf{(a)} represents the sensing protocol, including preparation unitary $U(\mathbf{p})$, the intermediate control $V$, and the transformation before readout $U(\mathbf{p})$. \textbf{(b)} shows all the corresponding values of each parameter in the circuit.}
        \label{fig:cnot-config}
    \end{figure}
    
    We can bound the bias as we have done in the case of SWAP interaction using the same technique as shown in \Cref{app:bound-bias-CE}.
    To obtain an upper bound of the overall error, we only need to calculate the phase factors in the unitaries shown in \Cref{fig:cnot-config}.
    To proceed, we first analyze the parameter circuit shown in \Cref{fig:vqa-ansatz}.
    The total phase in the circuit can be represented as
    \begin{align}
        \phi(U) \le \sum_{l} \max_i\{|\phi^{(l)}_i|\},
    \end{align}
    where indices $l$'s denote different layers and $i$'s represent the different gates in the same layer, and $\phi_i^{(l)}$ is the phase of a single gate.
    We assume the initial state is $\ket{0}$ in all qubits.
    By simple calculations, we can obtain that the phase in the pulses is bounded by
    \begin{align}
        |\phi_{\rm ctrl}|\le 4.366657005.
    \end{align}
    The phases accumulated in the free precession is bounded by
    \begin{align}
        |\phi_{\rm prec}|\le |\omega| T/2.
    \end{align}
    
    We conclude that the error in density matrix caused by the clock uncertainty is bounded by
    \begin{align}
        \Vert\tilde{\rho}-\rho\Vert \le 2 \epsilon (|\phi_{\rm prec}|+|\phi_{\rm ctrl}|),
    \end{align}
    where $\tilde{\cdot}$ denotes the term that affected by the error.
    The estimator can be considered as a function of the expectation value of the observable, which can be bounded by
    \begin{align}
        \left|\Tr[O_{\rm CE}(\tilde{\rho}-\rho)]\right|\le \Vert\tilde{\rho}-\rho\Vert\sum_{\lambda} |\lambda|,
    \end{align}
    where $\lambda$'s are eigenvalues of $O_{\rm CE}$.
    Then we can bound the bias in estimating $\omega$ by
    \begin{align}
        |b_{\rm CE}| &\le \left|\frac{\partial\hat{\omega}(x)}{\partial x}\right|\times \Vert\tilde{\rho}-\rho\Vert\sum_{\lambda}|\lambda|\\
        &=4\epsilon(|\phi_{\rm prec}|+|\phi_{\rm ctrl}|)\left|\frac{\partial\hat{\omega}(x)}{\partial x}\right|_{x=\braket{O_{\rm CE}}}\\
        &=\frac{8\epsilon}{T \sqrt{1-4 \braket{O_{\rm CE}}^2}}(|\phi_{\rm prec}|+|\phi_{\rm ctrl}|).
    \end{align}    
\end{widetext}

\end{document}